\documentclass[11pt,leqno, letterpaper]{article}
\usepackage[utf8]{inputenc}
\usepackage{dnaori}
\usepackage{gnuplot-lua-tikz}

\makeatletter
\newenvironment{proofw}{\par
  \pushQED{\qed}%
  \normalfont \topsep6\p@\@plus6\p@\relax
  \trivlist
  \item[]\ignorespaces
}{%
  \popQED\endtrivlist\@endpefalse
}
\makeatother

\global\long\def\AOS{\text{AOS}}
\global\long\def\ROS{\text{ROS}}
\global\long\def\IID{\text{i.i.d.}}

\title{Competitive Analysis with a Sample and the Secretary Problem}

\author{
Haim Kaplan\thanks{Blavatnik School of Computer Science, Tel Aviv University, Israel. Email:\texttt{haimk@tau.ac.il}. '
}
\and David Naori\thanks{Computer Science Department, Technion, Israel. Email:\texttt{\{dnaori,danny\}@cs.technion.ac.il}.} 
\and Danny Raz\footnotemark[2]
}
\date{}

\begin{document}
\maketitle
\begin{abstract}
We extend the standard online worst-case model to accommodate past experience which is available to the online player in many practical scenarios. We do this by revealing a random sample of the adversarial input to the online player ahead of time.
The online player competes with the expected optimal value on the part of the input that arrives online. Our model bridges between existing online stochastic models (e.g., items are drawn i.i.d.\ from a distribution) and the online worst-case model. We also extend in a similar manner (by revealing a sample) the online random-order model. 

We study the classical secretary problem in our new models. In the worst-case model we present a simple online algorithm with optimal competitive-ratio for any sample size. In the random-order model, we also give a simple online algorithm with an almost tight competitive-ratio for small sample sizes. Interestingly, we prove that for a large enough sample, no algorithm can be simultaneously optimal both in the worst-cast and random-order models. 

\end{abstract}
\section{Introduction}
Online algorithms have proven to be an important tool to study interactive scenarios where the input is revealed over time and we have to take decisions before seeing all the input. The analysis in the worst-case adversarial model provides robust guarantees 
of the expected performance when the algorithm has no prior knowledge about the input. However, in some cases this model is too powerful and no algorithm can achieve a non-trivial competitive-ratio (e.g., the secretary problem). In other cases (cache replacement policies for example) different algorithms provides similar worst-case guarantee, and thus the model does not help to distinguish between the algorithms~\cite{DBLP:journals/cacm/Roughgarden19}. Moreover, in many online scenarios, we have additional information about the online input (data from the past or other sources) and we want to use it.

To address this point we introduce a simple and natural generalization of the standard worst-case online model to accommodate past experience. As in this standard model, we allow an adversary to choose the input sequence as well as the order in which the input is revealed to the online player.
To model the experience that the online player may have, we assume she gets in advance a random sample of limited size from the adversarial input, which we call the \textit{history set}. The remaining part of the input, called the \textit{online set}, arrives online in adversarial order (that may depends on the random sample).
To evaluate the performance of an online algorithm, we compare the expected value of its solution to the expected optimal solution one could obtain in hindsight (i.e.,\ the optimal solution one can get if he knows the online set, and does not have to take decisions online. The expectation is over the random split into a history and online sets). We traditionally adopt the term \textit{competitive-ratio} to refer to the worst ratio between the two. The size of the random sample presented to the player is a parameter of the model; when limiting this size to be zero, the model is identical to the standard worst-case online model. We name this model the \textit{adversarial-order model with a sample ($\AOS$)}. We also consider the \textit{random-order model with a sample ($\ROS$)}, which is a similar generalization of the random-order online model, that is, a model in which the elements of the online set arrives in a uniformly random order. In this model the performance of the online algorithm is averaged not only on the split into history and online sets but also on the random order of the online set. 

Several models in the literature allow taking prior knowledge into consideration, however, they relay on overly strong assumptions that might not take place in most realistic conditions. 
(We overview notable examples in the related work subsection.) On the other hand, the random sample assumption we propose is focused solely on modeling the past experience the online player may have.

We apply our models to the classical online secretary problem which 
is arguably one of the most basic online problems. In the secretary problem, a sequence of candidates are presented one-by-one to an online player. Every candidate is associated with a value (non-negative real number) which is revealed when the candidate arrives. The online player may choose only one candidate, aiming to maximize the expected value of the candidate she chooses. 

In our $\AOS$ model, the adversary picks $n+h$ candidates. Then, $h$ uniformly random candidates are revealed to the online player upfront for the purpose of learning only. From here on, the process is identical to the (adversarial-order) secretary problem with the remaining $n$ candidates. 
Our results for this problem are illustrated by the two lower curves in Figure~\ref{fig:overview}. We distinguish between the cases where $h < n$ and $h \geq n$: For $h < n$ we describe a simple algorithm (Algorithm~\ref{alg:adv-sec-short-hist}) that achieves a competitive-ratio of $h/(n+h-1)$, and prove a matching upper-bound (Theorem~\ref{ao-upper-bound}). In particular, for $h=n-1$ the competitive-ratio of Algorithm~\ref{alg:adv-sec-short-hist} is $1/2$. For $h \geq n$, we show that this competitive-ratio of $1/2$ essentially cannot be further improved. We prove an upper-bound of $\frac{1}{2}\cdot \frac{2^n}{2^n-1}$ for this case (Theorem~\ref{thm:half_upper_bound}), and point-out a modification of Algorithm~\ref{alg:adv-sec-short-hist} (Algorithm~\ref{alg:adv-sec-long-hist}) that achieves a competitive-ratio of $1/2$ in this case ($h \geq n$).

We then move on to study the secretary problem in the $\ROS$ model. Our algorithm in this case can be viewed as a generalization of the well-known optimal algorithm for the ordinary secretary problem. The structure of the classical algorithm, sometimes referred to as \textit{sample-and-price}, is the same as the structure of most algorithms in the random-order model. It consists of two phases. In the first phase, known as the \textit{sampling phase}, the algorithm uses a prefix of the online input sequence to gather information about the input. This information is then used in the second phase to guide the decisions on the remaining part of the input. Taking the history set into account, our algorithm uses a sampling phase only when $h$ is not large enough (specifically, if $h\approx0.567n$ or larger, as illustrated in Figure~\ref{fig:overview}, a sampling phase is not used). More interestingly, it transition into a new phase when it accumulates a sample of size $n$. At this point it starts making decisions based on random subsets of the observed input. Our results for this case are illustrated by the two upper curves in Figure~\ref{fig:overview}. The competitive-ratio of our algorithm (Algorithm~\ref{alg:ro_secretary}) exhibits a trend similar to the $\AOS$ case. It improves as the size of the history set grows until $h=n$, and from there on it achieves a competitive-ratio of $1-(1-1/n)^n \geq 1-1/e$. We prove an upper-bound which is almost tight when $h$ is small compared to $n$, however, an interesting gap, especially for large $h$, remains unresolved.

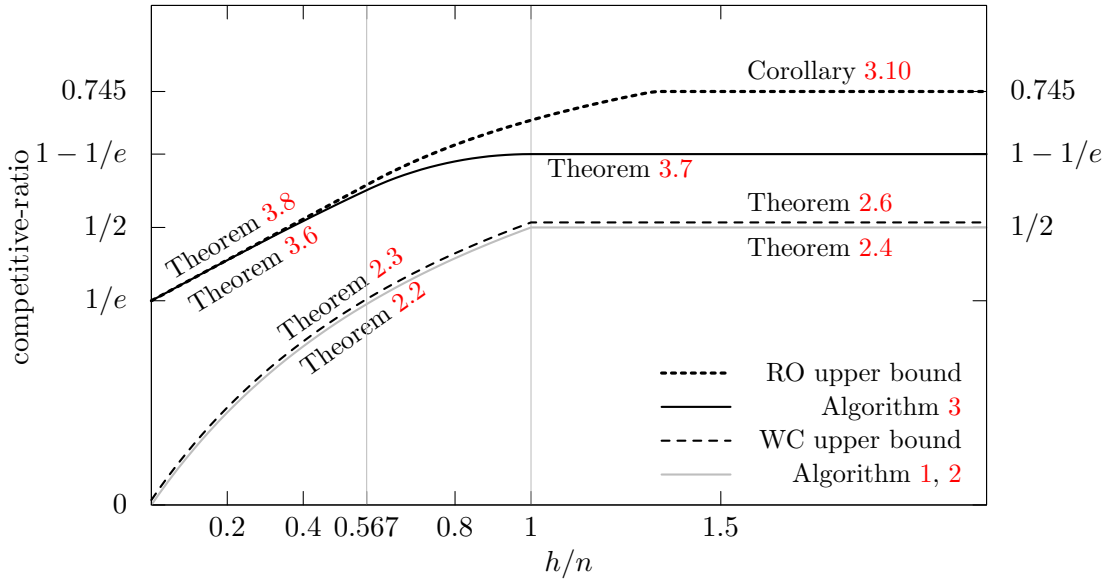
\begin{figure}
\centering
\caption{Overview of our results.}~\label{fig:overview}
\begin{tikzpicture}[gnuplot]
\path (0.000,0.000) rectangle (15.000,8.000);
\gpcolor{color=gp lt color border}
\gpsetlinetype{gp lt border}
\gpsetlinewidth{1.00}
\draw[gp path] (1.872,0.985)--(2.052,0.985);
\draw[gp path] (12.975,0.985)--(12.795,0.985);
\node[gp node right] at (1.688,0.985) {$0$};
\draw[gp path] (1.872,3.702)--(2.052,3.702);
\draw[gp path] (12.975,3.702)--(12.795,3.702);
\node[gp node right] at (1.688,3.702) {$1/e$};
\draw[gp path] (1.872,4.677)--(2.052,4.677);
\draw[gp path] (12.975,4.677)--(12.795,4.677);
\node[gp node right] at (1.688,4.677) {$1/2$};
\draw[gp path] (1.872,5.653)--(2.052,5.653);
\draw[gp path] (12.975,5.653)--(12.795,5.653);
\node[gp node right] at (1.688,5.653) {$1-1/e$};
\draw[gp path] (1.872,6.486)--(2.052,6.486);
\draw[gp path] (12.975,6.486)--(12.795,6.486);
\node[gp node right] at (1.688,6.486) {$0.745$};
\draw[gp path] (2.881,0.985)--(2.881,1.165);
\draw[gp path] (2.881,7.631)--(2.881,7.451);
\node[gp node center] at (2.881,0.677) {$0.2$};
\draw[gp path] (3.891,0.985)--(3.891,1.165);
\draw[gp path] (3.891,7.631)--(3.891,7.451);
\node[gp node center] at (3.891,0.677) {$0.4$};
\draw[gp path] (4.734,0.985)--(4.734,1.165);
\draw[gp path] (4.734,7.631)--(4.734,7.451);
\node[gp node center] at (4.734,0.677) {$0.567$};
\draw[gp path] (5.909,0.985)--(5.909,1.165);
\draw[gp path] (5.909,7.631)--(5.909,7.451);
\node[gp node center] at (5.909,0.677) {$0.8$};
\draw[gp path] (6.919,0.985)--(6.919,1.165);
\draw[gp path] (6.919,7.631)--(6.919,7.451);
\node[gp node center] at (6.919,0.677) {$1$};
\draw[gp path] (9.442,0.985)--(9.442,1.165);
\draw[gp path] (9.442,7.631)--(9.442,7.451);
\node[gp node center] at (9.442,0.677) {$1.5$};
\draw[gp path] (12.975,4.677)--(12.795,4.677);
\node[gp node left] at (13.159,4.677) {$1/2$};
\draw[gp path] (12.975,5.653)--(12.795,5.653);
\node[gp node left] at (13.159,5.653) {$1-1/e$};
\draw[gp path] (12.975,6.486)--(12.795,6.486);
\node[gp node left] at (13.159,6.486) {$0.745$};
\draw[gp path] (1.872,7.631)--(1.872,0.985)--(12.975,0.985)--(12.975,7.631)--cycle;
\node[gp node center,rotate=-270] at (0.15,4.308) {competitive-ratio};
\node[gp node center] at (7.423,0.215) {$h/n$};

\node[gp node left] at (9.659, 5) {\small{Theorem~\ref{thm:half_upper_bound}}};
\node[gp node left] at (9.659, 4.4) {\small{Theorem~\ref{thm_half_lower_bound}}};

\node[gp node left] at (7, 5.45) {\small{Theorem~\ref{thm:long_his_sec}}};
\node[gp node left] at (9.659, 6.75) {\small{Corollary~\ref{cor:hill_kertz}}};

\node[gp node left, rotate=31] at (2, 4) {\small{Theorem~\ref{ro-bound}}};
\node[gp node left, rotate=31] at (2.3, 3.55) {\small{Theorem~\ref{secretary CR}}};

\node[gp node left, rotate=35] at (3.5, 3.15) {\small{Theorem~\ref{ao-upper-bound}}};
\node[gp node left, rotate=35] at (3.8, 2.7) {\small{Theorem \ref{thm_lower_bound_short}}};

\gpcolor{rgb color={0.745,0.745,0.745}}
\gpsetlinetype{gp lt plot 0}
\draw[gp path](6.919,0.985)--(6.919,7.631);
\draw[gp path](4.734,0.985)--(4.734,7.631);
\gpcolor{color=gp lt color border}
\node[gp node right] at (12.791,2.740) {\small{RO upper bound}};
\gpcolor{rgb color={0.000,0.000,0.000}}
\gpsetlinetype{gp lt plot 2}
\gpsetlinewidth{3.00}
\draw[gp path] (8.655,2.740)--(9.571,2.740);
\draw[gp path] (1.872,3.702)--(1.984,3.762)--(2.096,3.822)--(2.208,3.883)--(2.321,3.943)%
  --(2.433,4.003)--(2.545,4.064)--(2.657,4.124)--(2.769,4.185)--(2.881,4.245)--(2.994,4.305)%
  --(3.106,4.366)--(3.218,4.426)--(3.330,4.486)--(3.442,4.547)--(3.554,4.607)--(3.666,4.667)%
  --(3.779,4.728)--(3.891,4.788)--(4.003,4.849)--(4.115,4.909)--(4.227,4.969)--(4.339,5.030)%
  --(4.451,5.090)--(4.564,5.150)--(4.676,5.211)--(4.788,5.271)--(4.900,5.331)--(5.012,5.388)%
  --(5.124,5.443)--(5.237,5.496)--(5.349,5.547)--(5.461,5.596)--(5.573,5.643)--(5.685,5.689)%
  --(5.797,5.733)--(5.909,5.776)--(6.022,5.817)--(6.134,5.857)--(6.246,5.895)--(6.358,5.933)%
  --(6.470,5.969)--(6.582,6.004)--(6.695,6.038)--(6.807,6.071)--(6.919,6.104)--(7.031,6.135)%
  --(7.143,6.165)--(7.255,6.195)--(7.367,6.223)--(7.480,6.251)--(7.592,6.279)--(7.704,6.305)%
  --(7.816,6.331)--(7.928,6.356)--(8.040,6.381)--(8.152,6.405)--(8.265,6.428)--(8.377,6.451)%
  --(8.489,6.473)--(8.601,6.486)--(8.713,6.486)--(8.825,6.486)--(8.938,6.486)--(9.050,6.486)%
  --(9.162,6.486)--(9.274,6.486)--(9.386,6.486)--(9.498,6.486)--(9.610,6.486)--(9.723,6.486)%
  --(9.835,6.486)--(9.947,6.486)--(10.059,6.486)--(10.171,6.486)--(10.283,6.486)--(10.396,6.486)%
  --(10.508,6.486)--(10.620,6.486)--(10.732,6.486)--(10.844,6.486)--(10.956,6.486)--(11.068,6.486)%
  --(11.181,6.486)--(11.293,6.486)--(11.405,6.486)--(11.517,6.486)--(11.629,6.486)--(11.741,6.486)%
  --(11.853,6.486)--(11.966,6.486)--(12.078,6.486)--(12.190,6.486)--(12.302,6.486)--(12.414,6.486)%
  --(12.526,6.486)--(12.639,6.486)--(12.751,6.486)--(12.863,6.486)--(12.975,6.486);
\gpcolor{color=gp lt color border}
\node[gp node right] at (12.791,2.290) {\small{Algorithm~\ref{alg:ro_secretary}}};
\gpcolor{rgb color={0.000,0.000,0.000}}
\gpsetlinetype{gp lt plot 0}
\gpsetlinewidth{2.00}
\draw[gp path] (8.655,2.290)--(9.571,2.290);
\draw[gp path] (1.872,3.702)--(1.984,3.762)--(2.096,3.822)--(2.208,3.883)--(2.321,3.943)%
  --(2.433,4.003)--(2.545,4.063)--(2.657,4.123)--(2.769,4.182)--(2.881,4.241)--(2.994,4.301)%
  --(3.106,4.359)--(3.218,4.418)--(3.330,4.476)--(3.442,4.534)--(3.554,4.592)--(3.666,4.649)%
  --(3.779,4.706)--(3.891,4.762)--(4.003,4.819)--(4.115,4.874)--(4.227,4.929)--(4.339,4.984)%
  --(4.451,5.038)--(4.564,5.092)--(4.676,5.145)--(4.788,5.198)--(4.900,5.247)--(5.012,5.294)%
  --(5.124,5.337)--(5.237,5.377)--(5.349,5.413)--(5.461,5.448)--(5.573,5.479)--(5.685,5.507)%
  --(5.797,5.533)--(5.909,5.556)--(6.022,5.577)--(6.134,5.595)--(6.246,5.610)--(6.358,5.624)%
  --(6.470,5.634)--(6.582,5.642)--(6.695,5.648)--(6.807,5.652)--(6.919,5.653)--(7.031,5.653)%
  --(7.143,5.653)--(7.255,5.653)--(7.367,5.653)--(7.480,5.653)--(7.592,5.653)--(7.704,5.653)%
  --(7.816,5.653)--(7.928,5.653)--(8.040,5.653)--(8.152,5.653)--(8.265,5.653)--(8.377,5.653)%
  --(8.489,5.653)--(8.601,5.653)--(8.713,5.653)--(8.825,5.653)--(8.938,5.653)--(9.050,5.653)%
  --(9.162,5.653)--(9.274,5.653)--(9.386,5.653)--(9.498,5.653)--(9.610,5.653)--(9.723,5.653)%
  --(9.835,5.653)--(9.947,5.653)--(10.059,5.653)--(10.171,5.653)--(10.283,5.653)--(10.396,5.653)%
  --(10.508,5.653)--(10.620,5.653)--(10.732,5.653)--(10.844,5.653)--(10.956,5.653)--(11.068,5.653)%
  --(11.181,5.653)--(11.293,5.653)--(11.405,5.653)--(11.517,5.653)--(11.629,5.653)--(11.741,5.653)%
  --(11.853,5.653)--(11.966,5.653)--(12.078,5.653)--(12.190,5.653)--(12.302,5.653)--(12.414,5.653)%
  --(12.526,5.653)--(12.639,5.653)--(12.751,5.653)--(12.863,5.653)--(12.975,5.653);
\gpcolor{color=gp lt color border}
\node[gp node right] at (12.791,1.840) {\small{WC upper bound}};
\gpcolor{rgb color={0.000,0.000,0.000}}
\gpsetlinetype{gp lt plot 1}
\draw[gp path] (8.655,1.840)--(9.571,1.840);
\draw[gp path] (1.872,1.051)--(1.984,1.212)--(2.096,1.366)--(2.208,1.513)--(2.321,1.654)%
  --(2.433,1.790)--(2.545,1.920)--(2.657,2.046)--(2.769,2.166)--(2.881,2.282)--(2.994,2.394)%
  --(3.106,2.502)--(3.218,2.606)--(3.330,2.707)--(3.442,2.804)--(3.554,2.898)--(3.666,2.988)%
  --(3.779,3.076)--(3.891,3.161)--(4.003,3.244)--(4.115,3.324)--(4.227,3.401)--(4.339,3.476)%
  --(4.451,3.549)--(4.564,3.620)--(4.676,3.689)--(4.788,3.756)--(4.900,3.821)--(5.012,3.884)%
  --(5.124,3.945)--(5.237,4.005)--(5.349,4.064)--(5.461,4.120)--(5.573,4.176)--(5.685,4.230)%
  --(5.797,4.282)--(5.909,4.333)--(6.022,4.383)--(6.134,4.432)--(6.246,4.480)--(6.358,4.526)%
  --(6.470,4.572)--(6.582,4.616)--(6.695,4.660)--(6.807,4.702)--(6.919,4.744)--(7.031,4.744)%
  --(7.143,4.744)--(7.255,4.744)--(7.367,4.744)--(7.480,4.744)--(7.592,4.744)--(7.704,4.744)%
  --(7.816,4.744)--(7.928,4.744)--(8.040,4.744)--(8.152,4.744)--(8.265,4.744)--(8.377,4.744)%
  --(8.489,4.744)--(8.601,4.744)--(8.713,4.744)--(8.825,4.744)--(8.938,4.744)--(9.050,4.744)%
  --(9.162,4.744)--(9.274,4.744)--(9.386,4.744)--(9.498,4.744)--(9.610,4.744)--(9.723,4.744)%
  --(9.835,4.744)--(9.947,4.744)--(10.059,4.744)--(10.171,4.744)--(10.283,4.744)--(10.396,4.744)%
  --(10.508,4.744)--(10.620,4.744)--(10.732,4.744)--(10.844,4.744)--(10.956,4.744)--(11.068,4.744)%
  --(11.181,4.744)--(11.293,4.744)--(11.405,4.744)--(11.517,4.744)--(11.629,4.744)--(11.741,4.744)%
  --(11.853,4.744)--(11.966,4.744)--(12.078,4.744)--(12.190,4.744)--(12.302,4.744)--(12.414,4.744)%
  --(12.526,4.744)--(12.639,4.744)--(12.751,4.744)--(12.863,4.744)--(12.975,4.744);
\gpcolor{color=gp lt color border}
\node[gp node right] at (12.791,1.390) {\small{Algorithm~\ref{alg:adv-sec-short-hist},~\ref{alg:adv-sec-long-hist}}};
\gpcolor{rgb color={0.745,0.745,0.745}}
\gpsetlinetype{gp lt plot 0}
\draw[gp path] (8.655,1.390)--(9.571,1.390);
\draw[gp path] (1.872,0.985)--(1.984,1.146)--(2.096,1.299)--(2.208,1.447)--(2.321,1.588)%
  --(2.433,1.723)--(2.545,1.854)--(2.657,1.979)--(2.769,2.100)--(2.881,2.216)--(2.994,2.328)%
  --(3.106,2.436)--(3.218,2.540)--(3.330,2.640)--(3.442,2.737)--(3.554,2.831)--(3.666,2.922)%
  --(3.779,3.010)--(3.891,3.095)--(4.003,3.177)--(4.115,3.257)--(4.227,3.335)--(4.339,3.410)%
  --(4.451,3.483)--(4.564,3.554)--(4.676,3.622)--(4.788,3.689)--(4.900,3.754)--(5.012,3.817)%
  --(5.124,3.879)--(5.237,3.939)--(5.349,3.997)--(5.461,4.054)--(5.573,4.109)--(5.685,4.163)%
  --(5.797,4.216)--(5.909,4.267)--(6.022,4.317)--(6.134,4.366)--(6.246,4.413)--(6.358,4.460)%
  --(6.470,4.505)--(6.582,4.550)--(6.695,4.593)--(6.807,4.636)--(6.919,4.677)--(7.031,4.677)%
  --(7.143,4.677)--(7.255,4.677)--(7.367,4.677)--(7.480,4.677)--(7.592,4.677)--(7.704,4.677)%
  --(7.816,4.677)--(7.928,4.677)--(8.040,4.677)--(8.152,4.677)--(8.265,4.677)--(8.377,4.677)%
  --(8.489,4.677)--(8.601,4.677)--(8.713,4.677)--(8.825,4.677)--(8.938,4.677)--(9.050,4.677)%
  --(9.162,4.677)--(9.274,4.677)--(9.386,4.677)--(9.498,4.677)--(9.610,4.677)--(9.723,4.677)%
  --(9.835,4.677)--(9.947,4.677)--(10.059,4.677)--(10.171,4.677)--(10.283,4.677)--(10.396,4.677)%
  --(10.508,4.677)--(10.620,4.677)--(10.732,4.677)--(10.844,4.677)--(10.956,4.677)--(11.068,4.677)%
  --(11.181,4.677)--(11.293,4.677)--(11.405,4.677)--(11.517,4.677)--(11.629,4.677)--(11.741,4.677)%
  --(11.853,4.677)--(11.966,4.677)--(12.078,4.677)--(12.190,4.677)--(12.302,4.677)--(12.414,4.677)%
  --(12.526,4.677)--(12.639,4.677)--(12.751,4.677)--(12.863,4.677)--(12.975,4.677);
\gpcolor{color=gp lt color border}
\gpsetlinetype{gp lt border}
\gpsetlinewidth{1.00}
\draw[gp path] (1.872,7.631)--(1.872,0.985)--(12.975,0.985)--(12.975,7.631)--cycle;
\gpdefrectangularnode{gp plot 1}{\pgfpoint{1.872cm}{0.985cm}}{\pgfpoint{12.975cm}{7.631cm}}
\end{tikzpicture}
\end{figure}

Our algorithm for the secretary problem in the $\ROS$ model improves upon recent results by Correa et al.~\cite{DBLP:conf/ec/CorreaDFS19} for the prophet inequality in the  $\IID$ model in which the online player gets access to a limited number of training samples from the (unknown) distribution. Intuitively, one can see that any algorithm for the secretary problem in the $\ROS$ model provides at least the same performance guarantee for the $\IID$ prophet inequality with a sample. Proving this simple observation also implies a (global) upper bound of approximately $0.745$ for the secretary problem in the $\ROS$ model (Corollary~\ref{cor:hill_kertz}).

After we study the secretary problem in each model separately, we explore how well can a single algorithm perform in both models simultaneously. Clearly, any performance guarantee in the $\AOS$ model also applies to the $\ROS$ model. However, we show that high worst-case performance guarantee might limit the increase in performance in the $\ROS$ model over the $\AOS$ model. More concretely, we show that a $c$-competitive algorithm in the $\AOS$ model, is at most $(1-c)$-competitive in the $\ROS$ model. For $h \geq n$, our algorithm for the $\AOS$ model (Algorithm~\ref{alg:adv-sec-long-hist}) is $1/2$-competitive and therefore cannot be more than $1/2$-competitive in the $\ROS$ model. In Section~\ref{sec:combine} we describe an algorithm that is simultaneously $1/e$-competitive in the $\AOS$ model and $(1-1/e)$-competitive in the $\ROS$ model.

Although in this paper we focus on the secretary problem, we stress that our models and methods are by no means limited to it. Azar et al.~\cite{soda/AzarKW14} observed that many online algorithms in the random-order model are in fact \textit{order-oblivious}, meaning that they use the random-order only to obtain a random sample from the input. In our models such a sample is given ``for free'' and we study how the online player should act when this is the case. In particular, our approach can be used to adapt existing order-oblivious algorithms to the $\AOS$ model and analyze their performance in this model.

In the $\ROS$ model, algorithms for various online problems can be obtained by combining our approach for the secretary problem with existing algorithms for online problems in the random-order model. For example, following the approach of Kesselheim et al.~\cite{DBLP:conf/esa/KesselheimRTV13}, our results for the secretary problem in the $\ROS$ model can be extended in a straightforward manner to the weighted bipartite matching problem, with the same performance guarantees.

Studying the $\AOS$ and the $\ROS$ models allows to distinguish between the power gained by the random-order assumption, and the ability to obtain a random sample from the input. Furthermore, these models allow for simple and direct analysis, as the proofs in this paper suggests.

\subsection{Further Related Work}
Since the secretary problem was solved by Lindley~\cite{lindley1961dynamic} and Dynkin~\cite{dynkin1963optimum}, various online problems have been studied in the random-order model (e.g. \cite{approx/BabaioffIKK07, soda/BabaioffIK07, DBLP:conf/esa/KesselheimRTV13, DBLP:conf/soda/Kleinberg05}), many of which are motivated by the relation to online mechanism design.  Kesselheim et al.~\cite{DBLP:conf/stoc/KesselheimKN15} studied the secretary problem with non-uniform arrival order, pointing out that in some cases, weakening the random-order assumption is essential.

There are several models in the literature where some prior knowledge about the input is assumed, however, in general these models introduce alongside additional assumptions that in many cases are not justifiable. One such model is the (known) $\IID$ model in which we assume the input consists of $\IID$ random variables from a known distribution (see~\cite{DBLP:conf/ec/CorreaDFS19, mehta2013online} for example). In most cases however, the knowledge of the exact distribution is unattainable and the model might not be robust in face of inaccurate estimates.

In the more general settings of the prophet inequality, the assumption that the random variables are identically distributed is discarded, and each random variable is allowed to be drawn from a different known distribution (for a recent survey see~\cite{DBLP:journals/sigecom/CorreaFHOV18}). The single sample prophet inequality (see~\cite{soda/AzarKW14} for example) relaxes the assumption that the distributions are known. Instead, the algorithm gets to sample one input sequence from the corresponding distributions for the purpose of learning. A main drawback of all these variations of the prophet inequality model is the strong assumption that the random variables are independent, which might not be realistic.

\subsection{Organization of the paper}
In Section~\ref{section:adv-order} we study the secretary problem in the $\AOS$ model. We begin by establishing formal definitions and notations, then we prove the lower and upper bounds for the case where $h < n$, and subsequently for the case where $h \geq n$. In Section~\ref{section:random-order} we study the secretary problem in the $\ROS$ model. Here too, we start with a formal definition of the problem in the $\ROS$ model, and prove lower and upper bounds. Then, we discuss the relations to the $\IID$ prophet inequality with a sample. Finally, in Section~\ref{sec:combine} we prove the inherent limitation of online algorithms when considered both in the $\AOS$ and the $\ROS$ models simultaneously, and discuss algorithms that achieve the best possible performance under this constraint.

\section{Adversarial Order}\label{section:adv-order}
We define the \textit{adversarial-order secretary problem with a sample of size $h$ ($h$-AO-SP)} as the following game between an online player and an adversary:
\begin{enumerate}
\item An adversary picks a set $\CS{C}=\left\{ \alpha_{1},\dots,\alpha_{n+h}\right\} $
of $n+h$ candidates. Each candidate $\alpha_{i}$ has a value $v\left(\alpha_{i}\right)\in\mathbb{R}_{\geq0}$.\footnote{Without loss of generality, one can think of the values as distinct. When this is not the case, we assume a consistent tie-breaker is available so that $\CS{C}$ is totally ordered. Throughout this paper, when candidates are compared by their value, we implicitly assume that this tie-breaker is applied.} For simplicity of notation, we use $\alpha_i$ to refer both the candidate and its value.

\item A subset $\RS{H} \subseteq \CS{C}$ of cardinality $h$, which we call the \textit{history set}, is drawn uniformly at random. $\RS{H}$ and $n$ are given to the online player upfront. The \textit{online set} denoted by $\RS{O}$ is the set of remaining candidates, i.e., $\RS{O}=\CS{C} \setminus \RS{H}$.
\item The adversary picks an ordering of the candidates in $\RS{O}$, we let $c_1,\dots, c_n$ denote the candidates in the chosen adversarial order. 
\item The candidates $c_1,\dots,c_n$ are presented one by one to the online player. After every arrival, the online player has to make an immediate and irrevocable decision whether to accept or reject the current candidate. If she accepts a candidate, the process terminates.
\end{enumerate}
The goal is to maximize the expected value of the accepted candidate compared to the expected value of the best candidate in $\RS{O}$. Let $\text{ALG}$ be an algorithm for the online player. For an instance $\mathcal{I}=\left(\CS{C},h\right)$, let $\text{ALG\ensuremath{\left(\mathcal{I}\right)}}$ be the random variable that gets the value of the candidate chosen by $\text{ALG}$, and let $\text{OPT\ensuremath{\left(\mathcal{I}\right)}}$ be the random variable that gets the maximum value of a candidate in $\RS{O}$. We say that $\ALG$ is \emph{$c$-competitive} if for every instance $\mathcal{I}$ we have $\E{ \text{ALG}\left(\mathcal{I}\right)}\geq c \cdot \E{\mbox{OPT}\left(\mathcal{I}\right)}$,
where the expectation is taken over the random choice of $\RS{H}\subseteq \CS{C}$ and the internal randomness of $\text{ALG}$. We write $\ALG$ and $\OPT$ instead of $\ALG(\mathcal{I})$ and $\OPT(\mathcal{I})$ when $\mathcal{I}$ is clear from the context.

\subsection{Short History}
For $h \leq n - 1$, we show that Algorithm~\ref{alg:adv-sec-short-hist} is optimal with competitive-ratio of $h/\left(n+h-1\right)$.
\begin{algorithm}
\label{alg:adv-sec-short-hist}
\caption{$h$-AO-SP for $h \leq n - 1$}
$T_{0}\leftarrow \RS{H}$\;

\For {candidate $c_{\ell}$ that arrives at round $\ell$ } {
    $T_{\ell} \leftarrow T_{\ell-1}\cup\left\{ c_{\ell}\right\}$\;
    \If {$c_{\ell}=\max\left\{T_\ell\right\}$ } {
	   accept $c_{\ell}$ and terminate\;
    }
}
\end{algorithm}

Let $\alpha_1,\dots,\alpha_{n+h}$ denote the candidates sorted by their value in decreasing order, i.e.,  $\alpha_1 > \alpha_2 > \dots > \alpha_{n+h}$. Observe that when $\alpha_1$ is in the online set and $\alpha_2$ is in the history set, the algorithm accepts $\alpha_1$ no matter what the adversary does. We start by extending this observation in a way that will allow us to account for the profit of the algorithm in case $\alpha_2$ is also in the online set. Denote by $\ES{H}_{i}$ the event that the top $i$ candidates are in the history set, that is, $\alpha_1,\dots,\alpha_i \in \RS{H}$. Also, denote by $\ES{O}_{i}$ the event that the top $i$ candidates are in the online set, namely, $\alpha_1,\dots,\alpha_i \in \RS{O}$.

\begin{lemma}~\label{Lemma adv}
For $1 < i \leq h \leq n-1$, $\E{\given{\ALG}{\ES{O}_{i}}} 
    \geq \frac{h}{n}\cdot\E{ \given{ \OPT }{ \ES{H}_{i-1} }}$.
\end{lemma}
\begin{proof}
By downwards induction on $i$. For $i=h$, conditioned on $\ES{H}_{h-1}$, the best candidate in $O$ has value of at most $\alpha_{h}$. On the other hand, observe that conditioned on $\ES{O}_{h}$ and $\alpha_{h+1} \in \RS{H}$, by the definition of the algorithm, it accepts one of $\alpha_{1},\dots,\alpha_{h}$, therefore, it gains a profit of at least $\alpha_{h}$. We have
\[
\E{\given{\ALG}{\ES{O}_{h}}} \geq \E{\given{\ALG}{\ES{O}_{h}, \alpha_{h+1} \in \RS{H}}}\Pr\left[ \given{\alpha_{h+1} \in \RS{H}} {\ES{O}_h} \right] = \frac{h}{n}\alpha_h.
\]
Assume that the lemma holds for $i+1$. We have
\begin{align}
\begin{split}
        \E{\given{\ALG}{\ES{O}_{i}}} 
        &\geq \alpha_{i} \Pr\left[\given{\alpha_{i+1} \in \RS{H}}{\ES{O}_{i}}\right] + \E{\given{\ALG}{\ES{O}_{i+1}}}
            \Pr\left[\given{\alpha_{i+1} \in \RS{O}}{\ES{O}_{i}}\right]\\  
        &\geq \alpha_{i} \Pr\left[\given{\alpha_{i+1} \in \RS{H}}{\ES{O}_{i}}\right] 
            + \frac{h}{n}\E{\given{\OPT}{\ES{H}_{i}}}
            \Pr\left[\given{\alpha_{i+1} \in \RS{O}}{\ES{O}_{i}}\right]\\
        &= \alpha_{i} \cdot \frac{h}{n+h-i} + \E{\given{\OPT}{\ES{H}_{i}}} \cdot \frac{h}{n}\cdot\frac{n-i}{n+h-i},
\end{split}~\label{eq:ALG_O_i}
\end{align}
where the second inequality in this derivation follows from the induction hypothesis. On the other hand, we have
\begin{align}
\begin{split}
    \E{ \given{ \OPT }{ \ES{H}_{i-1} }} 
    &= \alpha_{i} \Pr\left[\given{\alpha_{i} \in \RS{O}}{\ES{H}_{i-1}}\right] + \E{ \given{ \OPT }{ \ES{H}_{i} }}
         \Pr\left[\given{\alpha_{i} \in \RS{H}}{\ES{H}_{i-1}}\right] \\
    &= \alpha_{i} \cdot \frac{n}{n+h-(i-1)} + \E{ \given{ \OPT }{ \ES{H}_{i} }} \cdot \frac{h-(i-1)}{n+h-(i-1)}.
\end{split}~\label{eq:OPT_H_i-1}
\end{align}
By subtracting~\eqref{eq:OPT_H_i-1} multiplied by $h/n$ from~\eqref{eq:ALG_O_i},
the difference $\E{\given{\ALG}{\ES{O}_{i}}} - \frac{h}{n}\E{ \given{ \OPT }{ \ES{H}_{i-1} }}$ is lower bounded by
\begin{align*}
& \alpha_{i} 
\left( 
    \frac{h}{n+h-i} 
    - \frac{h}{n+h-i+1}
\right)
+
\frac{h}{n} \cdot \E{\given{\OPT}{\ES{H}_{i}}}
\left(
    \frac{n-i}{n+h-i} - \frac{h-i+1}{n+h-i+1}
\right).
\end{align*}
Since $h \leq n -1$, both terms are non-negative and the lemma follows.
\end{proof}
Having Lemma~\ref{Lemma adv} at hand, we are ready to prove the competitive-ratio of Algorithm~\ref{alg:adv-sec-short-hist}
\begin{theorem}\label{thm_lower_bound_short}
For $h\leq n-1$, Algorithm~\ref{alg:adv-sec-short-hist} is $\frac{h}{n+h-1}$-competitive.
\end{theorem}
\begin{proof}
Using Lemma~\ref{Lemma adv} for $i=2$, we get that
\begin{align*}
    \E{\ALG }
                &\geq \E{ \given{\ALG }{ \alpha_{1} \in \RS{O}, \alpha_{2} \in \RS{H} }}\Pr\left[ \alpha_{1} \in \RS{O}, \alpha_{2} \in \RS{H}\right] + \E{\given{\ALG}{ \alpha_{1}, \alpha_{2} \in \RS{O} }}\Pr\left[ \alpha_{1}, \alpha_{2} \in \RS{O}\right] \\ 
                &\geq  \frac{n}{n+h} \cdot \frac{h}{n+h-1}\alpha_{1} + \frac{n}{n+h} \cdot \frac{n-1}{n+h-1} \cdot \frac{h}{n}\E{ \given{ \OPT }{ \alpha_1 \in \RS{H} }}\\
                &\geq \frac{n}{n+h} \cdot \frac{h}{n+h-1}\left(  \alpha_{1} + \frac{h}{n}\E{ \given{ \OPT }{ \alpha_1 \in \RS{H} }}\right),
\end{align*}
where the last inequality follows from the assumption that $h \leq n-1$.
On the other hand
\begin{align*}
    \E{\OPT }
                &=  \E{ \given{ \OPT }{ \alpha_1 \in \RS{O} }}\Pr\left[ \alpha_{1} \in O\right] + \E{ \given{ \OPT }{ \alpha_1 \in \RS{H} }}\Pr\left[ \alpha_{1} \in \RS{H}\right] \\ 
                &= \frac{n}{n+h}\alpha_1 + \frac{h}{n+h}\E{ \given{ \OPT }{ \alpha_1 \in \RS{H} }}\\
                &= \frac{n}{n+h}\left(  \alpha_{1} + \frac{h}{n}\E{ \given{ \OPT }{ \alpha_1 \in \RS{H} }} \right). \qedhere
\end{align*}
\end{proof}
We now prove a matching upper bound.
\begin{theorem}~\label{ao-upper-bound}
For $h \geq 1$, any online algorithm for the $h$-AO-SP has a competitive-ratio of at most $\frac{h}{n+h-1}$. 
\end{theorem}
\begin{proof}
Fix $n \in \mathbb{N}$ and let $\varepsilon > 0$. Let $\ALG$ be a $c$-competitive algorithm for the $h$-AO-SP. We can view $\ALG$ restricted to inputs of length $n+h$ as a family of functions, $P_1,\dots,P_n$, $P_i:\mathbb{R}_{\geq 0}^{h+i} \rightarrow [0,1]$ where $P_i(x_1,\dots,x_{h+i})$ is the probability that $\ALG$ accepts $x_{h+i}$ conditioned on reaching to round $i$ and receiving $H=\left\{x_1,\dots,x_h\right\}$ and $c_1=x_{h+1},\dots,c_{i}=x_{h+i}$ as input. Since $\ALG$ receives the elements of $H$ unordered, we may assume that the first $h$ inputs to $P_i$ are ordered in increasing order. We call two sequences $x_1,\dots, x_k$ and $y_1,\dots,y_k$ \textit{order-equivalent} if for all $1\leq i,j \leq k$, $x_i < x_j \iff y_i < y_j$. We call the equivalence class of $x_1,\dots,x_k$ its \textit{order-type}. We say that $\ALG$ is \textit{order-invariant} on a set $\CS{V}$ if for all $i\in [n]$, the value of $P_i$ on a sequence of $h+i$ elements from $\CS{V}$ depends only on the order-type of the sequence.

By Moran et al.~\cite{DBLP:journals/jacm/MoranSM85} (Corollary 3.4), there is an infinite set $\CS{V} \subseteq \mathbb{N}$ such that $\ALG$ is order-invariant on $\CS{V}$. We construct an instance $\mathcal{I}=(\CS{C},h)$ where $\CS{C}= \left\{\alpha_1,\dots, \alpha_{n+h}\right\} \subseteq \CS{V}$, $\alpha_2 > \alpha_3 > \dots > \alpha_{n+h}$ and $\alpha_1 > \alpha_2 / \varepsilon'$ where $\varepsilon' = \frac{n}{n+h} \varepsilon$. Such instance exists since $\CS{V} \subseteq \mathbb{N}$ is infinite. Now every choice of an element that is not $\alpha_1$ results in a profit of at most $\varepsilon' \alpha_1$. Let $\RS{U}$ denote the set of candidates with value higher than the best candidate in $\RS{H}$. Consider an adversary who first reveals the elements of $\RS{U}$ in increasing order. Observe that for all $i \in [n]$, we have $c_i = \alpha_1$ if and only if $\alpha_1, \dots, \alpha_i \in O$ and $\alpha_{i+1} \in H$, therefore
\begin{align}
\Pr[c_i = \alpha_1] = \frac{h}{n+h}\cdot \frac{n}{n+h-1} \cdot \frac{n-1}{n+h-2} \cdots \frac{n-(i-1)}{n+h-i} \leq \frac{h}{n+h}\cdot \frac{n}{n+h-1}.~\label{eq:p_ci_is_ai}
\end{align}
By our assumption that the elements of $H$ are ordered in increasing order, the entire prefix of the observed sequence until $\alpha_1$ is in increasing order, hence, it is order-equivalent to $\alpha_{i+h},\dots,\alpha_1$ for some $i \in [n]$. Since $\ALG$ is order-invariant on $\CS{V}$, the probability of accepting $c_j$ conditioned on reaching round $j$ and on $c_i = \alpha_1$ for some $i \geq j$ is $p_j = P_j(\alpha_{j+h}, \dots, \alpha_1)$. Therefore, conditioned on $c_i = \alpha_1$, the probability of reaching round $i$ is $\prod_{j=1}^{i-1} \left(1-p_j\right)$ independently of $c_1,\dots,c_{i-1}$. We get that for all $i\in [n]$
\begin{align}
    \Pr \left[ \given{ \ALG \text{ accepts } \alpha_1}{c_i = \alpha_1} \right] 
    = p_i \prod_{j=1}^{i-1} (1 - p_j).~\label{eq:p_accept_ci_ai}
\end{align}
Hence, 
\begin{align}
\begin{split}
\Pr\left[\ALG \text{ accepts } \alpha_1\right] &=
\sum_{i=1}^{n} \Pr \left[\given{\ALG \text{ accepts } \alpha_1}{c_i = \alpha_1} \right] \cdot \Pr\left[c_i = \alpha_1\right] \\
&\leq \sum_{i=1}^{n} p_i \prod_{j=1}^{i-1} (1 - p_j)\cdot \frac{h}{n+h}\cdot \frac{n}{n+h-1} \\
&= \frac{h}{n+h}\cdot \frac{n}{n+h-1}\sum_{i=1}^{n} p_i \prod_{j=1}^{i-1} (1 - p_j)
\leq \frac{h}{n+h}\cdot \frac{n}{n+h-1}.
\end{split} \notag
\end{align}
The first inequality in this derivation follows from Equations~\eqref{eq:p_ci_is_ai} and~\eqref{eq:p_accept_ci_ai}, and the second inequality is due to the fact that $\sum_{i=1}^{n} p_i \prod_{j=1}^{i-1} (1 - p_j)$ is a probability of some event, and as such, upper bounded by $1$. We get that
\begin{align*}
    \E{\ALG} 
    &\leq \alpha_1 \Pr[\ALG \text{ accepts } \alpha_1] + \varepsilon' \alpha_1 \\
    &\leq \alpha_1 \left(  \frac{h}{n+h}\cdot \frac{n}{n+h-1} + \varepsilon' \right) \\
    &= \alpha_1 \frac{n}{n+h} \left( \frac{h}{n+h-1} + \varepsilon \right),
\end{align*}
while $\E{\OPT} \geq \alpha_1 \frac{n}{n+h}$. Overall $\E{\ALG} \leq \left(\frac{h}{n+h-1} + \varepsilon \right)\E{\OPT}$. Since it is true for any $\varepsilon>0$, the theorem follows.

\end{proof}

\subsection{Long History}
For the case $h \geq n$, we describe an optimal $1/2$-competitive algorithm.
\begin{algorithm}
\label{alg:adv-sec-long-hist}
\caption{$h$-AO-SP for $h \geq n$}
draw a subset $\RS{T} \subseteq \RS{H}$ of cardinality $n-1$ uniformly at random\;
\For {candidate $c_{\ell}$ that arrives at round $\ell$ } {
    \If {$c_{\ell} > \max\left\{\RS{T}\right\}$ } {
	   accept $c_{\ell}$ and terminate\;
    }
}
\end{algorithm}

\begin{theorem}~\label{thm_half_lower_bound}
Algorithm~\ref{alg:adv-sec-long-hist} is $1/2$-competitive.
\end{theorem}
\begin{proof}
For the analysis, we think of the selection of the set $\RS{H} \subseteq \CS{C}$ as being determined by the following process: first, a subset $\RS{U} \subseteq \CS{C}$ of cardinality $2n-1$ is chosen uniformly at random, then the online set $\RS{O} \subseteq \RS{U}$ of cardinality $n$ is chosen uniformly at random, and $\RS{H} = \CS{C} \setminus \RS{O}$. Since $\RS{U}$ is a uniformly random subset of cardinality $2n-1$, $\RS{U}$ and $\RS{T} \cup \RS{O}$ are identically distributed. Fix $\RS{U} = \CS{Y}$. Let $\beta_{1}, \dots, \beta_{2n-1}$ denote the candidates in $\CS{Y}$ ordered by their value in decreasing order. We have
\begin{align}
\begin{split}
\E{\given{\OPT}{\RS{U}=\CS{Y}}} &=
\beta_1 \Pr\left[\given{\beta_1 \in \RS{O}}{\RS{U}=\CS{Y}}\right]+ \E{\given{\OPT}{\beta_1 \notin \RS{O}, \RS{U}=\CS{Y}}}\Pr\left[\given{\beta_1 \notin \RS{O}}{\RS{U}=\CS{Y}}\right] \\
&=  \beta_1 \frac{n}{2n-1} + \E{\given{\OPT}{\beta_1 \notin \RS{O}, \RS{U}=\CS{Y}}} \frac{n-1}{2n-1}.
\end{split}~\label{eq:max_o} 
\end{align}
On the other hand
\begin{align}
\E{\given{\ALG}{\RS{T} \cup \RS{O} = \CS{Y}}} 
&\geq \beta_1\Pr\left[\given{\beta_1 \in \RS{O}, \beta_2 \in \RS{T}}{ \RS{T} \cup \RS{O} = \CS{Y}}\right] \notag \\
&\quad + \E{\given{\ALG}{\beta_1,\beta_2 \in \RS{O}, \RS{T} \cup \RS{O} = \CS{Y}}}\Pr\left[\given{\beta_1,\beta_2 \in \RS{O}}{\RS{T}\cup \RS{O} = \CS{Y}}\right] \notag \\
&= \beta_{1} \frac{n}{2n-1}\cdot\frac{n-1}{2n-2}  + \E{\given{\ALG}{\beta_1,\beta_2 \in \RS{O}, \RS{T} \cup \RS{O} = \CS{Y}}} \frac{n}{2n-1}\cdot\frac{n-1}{2n-2}.~\label{eq:alg_t_o}
\end{align}
Following the proof of Lemma~\ref{Lemma adv}, we have
\begin{align}
\E{\given{\ALG}{\beta_1,\beta_2 \in \RS{O}, \RS{T} \cup \RS{O} = \CS{Y}}} \geq \frac{n-1}{n}\E{\given{\OPT}{\beta_1 \notin O, U=\CS{Y}}}.\label{eq:alg_given}
\end{align}
Substituting \eqref{eq:alg_given} in Inequality \eqref{eq:alg_t_o} we get
\begin{align}
\begin{split}
\E{\given{\ALG}{\RS{T} \cup \RS{O} = \CS{Y}}} &\geq
\beta_{1} \cdot \frac{n}{2n-1} \cdot \frac{n-1}{2n-2}  + \E{\given{\OPT}{\beta_1 \notin O, U=\CS{Y}}} \cdot \frac{n-1}{2n-1} \cdot \frac{n-1}{2n-2} \\
&= \frac{1}{2}\E{\given{\OPT}{\RS{U}=\CS{Y}}}\label{eq:half_comp},
\end{split}
\end{align}
where the last equality follows from~\eqref{eq:max_o}. By law of total expectation, we get 
\begin{align*}
\E{\ALG} &= \sum_{\substack{\CS{Y}\subseteq \CS{C} \\ |\CS{Y}| = 2n-1}}{\E{\given{\ALG}{\RS{T} \cup \RS{O} = \CS{Y}}}\Pr\left[ \RS{T}\cup \RS{O} = \CS{Y}\right]} \\
&\geq  \sum_{\substack{\CS{Y}\subseteq \CS{C} \\ |\CS{Y}| = 2n-1}} {\frac{1}{2}\E{\given{\OPT}{\RS{U}=\CS{Y}}}\Pr\left[ \RS{U} = \CS{Y}\right]} \\
&= \frac{1}{2} \E{\OPT},
\end{align*}
where the inequality follows from~\eqref{eq:half_comp} and the fact that $\Pr\left[ \RS{T}\cup \RS{O} = \CS{Y}\right] = \Pr\left[ \RS{U} = \CS{Y}\right]$.
\end{proof}

Next we prove that asymptotically, no online algorithm can be better than $1/2$-competitive. To this end we use the following proposition.

\begin{proposition}~\label{prop:choose-r}
For any $n,k \in \mathbb{N}$ such that $n \geq k$, and $r \geq 1$ such that $r n \in \mathbb{N}$, we have $\binom{n}{k} \leq \frac{1}{r^k}\binom{r n}{k} $.
\end{proposition}

For a proof see Appendix~\ref{apx:choose-r}.

\begin{theorem}~\label{thm:half_upper_bound}
Any online algorithm for the $h$-AO-SP has a competitive-ratio of at most $\frac{1}{2} \cdot \frac{2^n}{2^n -1}$.
\end{theorem}
\begin{proof}
Let $\ALG$ be a $c$-competitive algorithm. Fix $\varepsilon > 0$. We construct two instances: the first instance $\mathcal{I}_1$ consists of a set $\CS{E}$ of $(h+n)/2$ candidates of value $\varepsilon$, and a set $\CS{Z}$ of $(h+n)/2$ candidates of value $0$.\footnote{We assume, without loss of generality, that $h+n$ is even.} For the second instance $\mathcal{I}_2$, we replace one arbitrary candidate of value $0$ with a candidate $\alpha$ of value $\frac{n+h}{n}$. Consider an adversary who first reveals all candidates in $\CS{E} \cap \RS{O}$, then, all candidates in $\CS{Z} \cap \RS{O}$. At the end, it reveals $\alpha$ if $\alpha \in \RS{O}$. Let $p$ be the probability that at least one candidate from $\CS{Z}$ is in the online set. We have $p = \Pr[\CS{Z} \cap \RS{O} \neq \emptyset] = \Pr[\CS{E} \cap \RS{O} \neq \emptyset]$. By Proposition~\ref{prop:choose-r} for $r=2$, we have $p \geq 1-1/2^n$.

Let us first consider $\mathcal{I}_1$. We have $\E{\OPT(\mathcal{I}_1)} = \varepsilon \cdot p$. Now let $\ES{A}(\mathcal{I})$ denote the event that $\ALG(\mathcal{I})$ accepts a candidate of value $\varepsilon$.
By the assumption that $\ALG$ is $c$-competitive, on $\mathcal{I}_1$ it must accept a candidate of value $\varepsilon$ with probability $c \cdot p$, that is, $\Pr \left[\ES{A}(\mathcal{I}_1)\right] \geq c\cdot p$. Therefore
\begin{align}
 c \cdot p
 &\leq \Pr \left[\ES{A}(\mathcal{I}_1)\right]
 \leq  \Pr \left[ \given{\ES{A}(\mathcal{I}_1)} {\CS{Z} \cap O \neq \emptyset} \right]\cdot p + (1 - p).~\label{eq:c-p-epsilon}
\end{align}

For the second instance, we have $  \E{\OPT(\mathcal{I}_2)} \geq \alpha \Pr[\alpha \in O] = \frac{n+h}{n} \cdot \frac{n}{n+h} = 1$. Since $\ALG$ is $c$-competitive, we have $\E{\ALG\left(\mathcal{I}_2\right)} \geq c \cdot \E{\OPT\left(\mathcal{I}_2\right)} \geq c$. On the other hand, when $\alpha \in \RS{H}$ the profit is at most $\varepsilon$. We get that
\begin{align*}
    \E{\ALG\left(\mathcal{I}_2\right)} &\leq \alpha \Pr[\ALG\text{ accepts } \alpha \wedge \alpha \in O] + \varepsilon \\
    &= \alpha \Pr[\alpha \in O]\Pr\left[\given{\ALG\text{ accepts } \alpha}{\alpha \in O}\right]  + \varepsilon \\
    &= \Pr\left[\given{\ALG\text{ accepts } \alpha}{\alpha \in O}\right] + \varepsilon.
\end{align*}
Therefore, conditioned on $\alpha \in \RS{O}$, the probability that $\ALG$ accepts $\alpha$ is at least $c-\varepsilon$. 
Since all the candidates in $\CS{E}$ arrive first and $\alpha$ arrives last, by Inequality~\eqref{eq:c-p-epsilon} and the lower bound on $p$ we have
\begin{align*}
\Pr \left[ \given{A(\mathcal{I}_2)}{\alpha \in \RS{O}} \right] 
&= \Pr\left[ \given{A(\mathcal{I}_1)}{\CS{Z} \cap O \neq \emptyset} \right] 
\geq \frac{c \cdot p - (1-p) }{ p } \geq c - \frac{2^n}{2^n - 1} + 1.
\end{align*}
Conditioned on $\alpha \in \RS{O}$, the probability that $\ALG$ accepts $\alpha$ is at most $1- \Pr\left[\given{A(\mathcal{I}_2)}{\alpha \in \RS{O}}\right]$. We get that
\begin{align*}
c - \varepsilon \leq 1- \Pr\left[\given{A(\mathcal{I}_2)}{\alpha \in \RS{O}}\right] \leq \frac{2^n}{2^n - 1} - c.
\end{align*}
Hence, $c \leq \frac{1}{2} \cdot \frac{2^n}{2^n-1} + \frac{\varepsilon}{2}$. 
\end{proof}

\section{Random Order}\label{section:random-order}
We define the \textit{random-order secretary problem with a sample of size $h$ ($h$-RO-SP)} similarly to the game defined for the adversarial order case, the only difference is that the order in which the candidates in $\RS{O}$ are presented to the online player is chosen uniformly at random. More explicitly, we define it as the following game:

\begin{enumerate}
\item An adversary picks a set $\CS{C}=\left\{ \alpha_{1},\dots,\alpha_{n+h}\right\} $
of $n+h$ candidates. Each candidate $\alpha_{i}$ has a value $v\left(\alpha_{i}\right)\in\mathbb{R}_{\geq0}$.\footnote{The assumptions we made regarding the values of the candidates in Section~\ref{section:adv-order} also apply in this case.}

\item A subset $\RS{H} \subseteq \CS{C}$ of cardinality $h$ is drawn uniformly at random.

\item The candidates in $\RS{O}$ are presented to the online player one by one in a uniformly random order. After every arrival, the online player has to make an immediate and irrevocable decision whether to accept or reject the current candidate. If she accepts a candidate, the process terminates.

\end{enumerate}
Let $\text{ALG}$ be an algorithm for the online player. In this section, $\E{\ALG}$ is taken over the random choice of $\RS{H}\subseteq \CS{C}$, the random arrival order of candidates in $\RS{O}$ and the internal randomness of $\ALG$.

We consider a natural algorithm for this problem. It operates in three phases which we call the \textit{sampling phase}, \textit{Phase~\ref{phase 1}} and \textit{Phase~\ref{phase 2}}. The sampling phase and Phase 1, are similar to the optimal algorithm for the ordinary secretary problem. The parameter $q$ is the fraction of rounds used for the sampling phase. It is determined as a function of $n$ and $h$ which emerges from the analysis.\footnote{We assume $qn$ is an integer.} In Phase~\ref{phase 1} the algorithm accepts a candidate if he is the best so far (including the history set). In Phase~\ref{phase 2}, at every round the algorithm selects $n-1$ random candidates from the past, and accepts the current candidate if he is the best among them.
\IncMargin{1em}
\begin{algorithm}
\label{alg:ro_secretary}
\caption{$h$-RO-SP}
$T_{0}\leftarrow \RS{H}$\;
$q \leftarrow \max\left\{ e^{-e^{-h/n}} - \frac{h}{n} ,0\right\}$\;
\For {candidate $c_{\ell}$ that arrives at round $\ell$ } {
    $T_{\ell} \leftarrow T_{\ell-1}\cup\left\{ c_{\ell}\right\}$\;
    \uIf(\tcc*[f]{sampling phase}) {$ \ell \leq q n$} {
        continue to the next round\;
    }
    \lnl{phase 1}\uElseIf(\tcc*[f]{phase 1}){ $\left| T_{\ell} \right| \leq n$} {
        \If {$c_{\ell}=\max\left\{ T \right\}$} {
	        accept $c_{\ell}$ and terminate\;
        }
    }
    \lnl{phase 2}\Else(\tcc*[f]{phase 2}) {
        draw a subset $\RS{X}_{\ell}\subseteq T_{\ell - 1}$ of cardinality $(n-1)$ uniformly at random\;
        \If {$c_{\ell}>\max\left\{\RS{X}_{\ell}\right\}$} {
	        accept $c_{\ell}$ and terminate\;
        }
    }
}
\end{algorithm}
\DecMargin{1em}

We now analyze the performance of Algorithm~\ref{alg:ro_secretary}. Unless specifically indicated otherwise
we assume that $h \leq n-1$. Our reasoning still applies for $h \geq n$, but since Phase~\ref{phase 1} is skipped completely in this case, it requires a different base-case which we discuss afterwards. We bound the expected profit of the algorithm at each round separately. To this end, for a fixed round $\ell$, we think of the random process that leads to round $\ell$ as if it is determined by the following steps:
\begin{enumerate}[label=({\arabic*})]
    \item~\label{step 1} First, a set of candidates $\RS{S}_\ell \subseteq \CS{C}$ of cardinality $h + \ell$ is chosen uniformly at random.
    \item~\label{step 2} Second, the candidate to arrive at round $\ell$, $c_\ell$, is chosen uniformly at random from $S_\ell$. Let $\RS{S}_{\ell-1}\leftarrow \RS{S}_{\ell} \setminus \left\{ c_\ell \right\} $.
    \item~\label{step 3} Finally, Step~\ref{step 2} is repeated with $\RS{S}_{\ell - 1}$ to determine the candidates that arrive at rounds $\ell - 1,\dots,1$. $\RS{H}$ is the set of the remaining candidates, i.e., $\RS{S}_{0}$.
\end{enumerate}

We now bound the expected profit of the algorithm at each round separately. Let $\RS{R}_\ell$ denote the profit of the algorithm at round $\ell$. In Lemma~\ref{Not Accepted Phase1} we derive the probability that $\ALG$ rejects the first $\ell$ candidates for a round $\ell$ in Phase~\ref{phase 1}. Then, we use this result to bound $\E{R_\ell}$ in Lemma~\ref{Secretary Lemma}. Analogous results for Phase~\ref{phase 2} are presented in Lemma~\ref{Not Accepted Phase2} and Lemma~\ref{Secretary History Lemma}. 

For $qn + 1 \leq k \leq n$, we denote by $\ES{M}_k$ the event that $c_k$ is accepted by the algorithm.
\begin{lemma}~\label{Not Accepted Phase1}
For $qn + 1 \leq \ell \leq n-h$, for any  $\CS{U}_\ell \subseteq \CS{C}$ such that $\left| \CS{U}_\ell \right|= h+\ell$, we have
\begin{align*}
 \Pr \left[\given{\bigwedge_{k=qn+1}^{\ell} \neg \ES{M}_k}{\RS{S}_{\ell}=\CS{U}_\ell}\right] = \frac{h+qn}{h+\ell}.
\end{align*}
\end{lemma}
\begin{proof}
Recall that $\RS{T}_{qn}$ (defined in the description of Algorithm~\ref{alg:adv-sec-short-hist}) is the set of $h+qn$ candidates that the algorithm observes before Phase~\ref{phase 1} begins. Conditioned on $S_\ell = \CS{U}_\ell$, by the definition of Phase~\ref{phase 1}, if $\max{\left\{\CS{U}_\ell\right\}}$ appears in $\RS{T}_{qn}$, the algorithm rejects all candidates until round $\ell$. Conversely, if  $\max{\left\{\CS{U}_\ell\right\}} \notin \RS{T}_{qn}$, the algorithm must accept it, or some other candidate before it encounters  $\max{\left\{\CS{U}_\ell\right\}}$. We get
\[
 \Pr \left[\given{\bigwedge_{k=qn+1}^{\ell} \neg \ES{M}_k}{\RS{S_\ell}=\CS{U}_\ell}\right] = \Pr\left[\given{ \max{\left\{\CS{U}_\ell\right\}} \in T_{qn}}{\RS{S}_\ell = \CS{U}_\ell}\right]
 = \frac{h+qn}{h+\ell}.
 \qedhere
\]
\end{proof}

For the proof of Lemma~\ref{Secretary Lemma} below, we also use the following proposition
\begin{proposition}~\label{expected max prop}
Suppose $\CS{X} \subseteq \mathbb{R}$ is a finite set. Let $\RS{A} \subseteq \CS{X}$ and $\RS{B} \subseteq \CS{X}$ be uniformly random subsets of cardinality $k$ and $n$ respectively, where $k \leq n$. Then, $\E{\max\left\{ \RS{A} \right\}} \geq \frac{k}{n} \E{\max\left\{\RS{B}\right\}}$.
\end{proposition}
For a proof see Appendix~\ref{proof expected max prop}.

\begin{lemma}~\label{Secretary Lemma}
For $  q  n + 1\leq \ell \leq n  - h $, we have 
\[ 
\E{\RS{R}_\ell} = \frac{h+qn}{h+\ell-1}\cdot\frac{1}{n}\E{\OPT}.
\]
\end{lemma}

\begin{proof}
By Step~\ref{step 1}, $\RS{S}_\ell \subseteq \CS{C}$ is a uniformly random subset of size $h + \ell \leq n$, therefore, by Proposition~\ref{expected max prop}, $\E{\max\left\{\RS{S}_\ell\right\} } \geq \frac{h + \ell}{n} \E{\OPT}$. We now bound $\E{\given{R_\ell}{S_\ell = \CS{U_\ell}}}$. We have
\begin{align*}
    \E{\given{R_\ell}{S_\ell = \CS{U}_\ell}} &= \max\left\{\CS{U}_\ell\right\} \Pr \left[\given{\ES{M}_\ell \bigwedge_{k=qn+1}^{\ell-1} \neg \ES{M}_k}{S_\ell = \CS{U_\ell}}\right]\\ 
    &= \max\left\{\CS{U}_\ell\right\} \Pr \left[\given{\bigwedge_{k=qn+1}^{\ell-1} \neg \ES{M}_k}{S_{\ell-1} = \CS{U}_\ell \setminus \left\{\max\left\{\CS{U_\ell}\right\}\right\} }\right] \frac{1}{h + \ell} \\
    &= \max\left\{\CS{U}_\ell\right\} \frac{h+qn}{h+\ell -1 } \cdot \frac{1}{h + \ell},
\end{align*}
where the first equality is due to the fact that $c_\ell$ is accepted if and only if $c_\ell = \max\left\{\CS{U}_\ell\right\}$, the second equality follows from the fact that $\Pr\left[\given{c_\ell = \max\left\{\CS{U}_\ell\right\}}{\RS{S}_\ell = \CS{U}_\ell}\right] = \frac{1}{h+\ell}$ due to Step~\ref{step 2}, and the third equality follows from Lemma~\ref{Not Accepted Phase1}. By law of total expectation, we get
\begin{align*}
    \E{R_\ell} &=  \E{\E{\given{R_\ell}{S_\ell}}} = \E{\max\left\{S_\ell\right\} \frac{h+qn}{h+\ell -1 } \cdot \frac{1}{h + \ell}} = \frac{h+qn}{h+\ell-1}\cdot\frac{1}{n}\E{\OPT}. \qedhere
\end{align*}
\end{proof}
We now move to bound the expected profit of the algorithm at Phase~\ref{phase 2}.
\begin{lemma}\label{Not Accepted Phase2}
For $n-h \leq \ell \leq n$, for any $\CS{U}_\ell\subseteq \CS{C}$ such that $\left|\CS{U}_\ell\right|= h+\ell$, we have
\begin{align*}
 \Pr \left[\given{\bigwedge_{k=qn+1}^{\ell} \neg \ES{M}_k}{\RS{S_\ell}=\CS{U}_\ell}\right] = \frac{h+qn}{n} \cdot \left(1-\frac{1}{n} \right)^{\ell-(n-h)}
\end{align*}
\end{lemma}
\begin{proof}
For $\ell=n-h$, by Lemma~\ref{Not Accepted Phase1}, the claim holds. Now assume by induction that the claim holds for $\ell - 1$. Fix $\CS{U}_\ell \subseteq \CS{C}$. Let $\FS{F}_m$ be the family of subsets of $\CS{U}_\ell$ of cardinality $m$. For $\alpha \in \CS{U}_\ell$ let $\FS{F}_{m}^\alpha \subseteq \FS{F}_{m}$ be the family of subsets in which the maximum value is greater than $\alpha$. Recall that $X_\ell$ is the subset drawn by $\ALG$ at round $\ell$. For the purpose of the analysis, we assume that $\ALG$ draws all of $X_1,\dots,X_n$, regardless of the round in which it picks the secretary. Observe that conditioned on $S_{\ell} = \CS{U}_\ell$ and $c_\ell = \alpha$, $\ALG$ rejects $c_\ell$ if and only if $X_\ell \in \FS{F}_{n-1}^\alpha$. By this observation along with the law of total probability, we have
\begin{align}
\begin{split}
\Pr \left[\given{\bigwedge_{k=qn+1}^{\ell} \neg \ES{M}_\ell}{\RS{S}_\ell = \CS{U}_\ell}\right]
 &= \sum_{\alpha \in \CS{U}_\ell} \sum_{\CS{X} \in \FS{F}^{\alpha}_{n-1}}
    \Pr \left[\given{\bigwedge_{k=qn+1}^{\ell-1} \neg \ES{M}_\ell}{ \RS{S}_{\ell}=\CS{U}_\ell, c_\ell = \alpha, X_\ell = \CS{X} } \right]\\
    & \qquad \qquad \qquad \cdot \Pr\left[ \given{c_\ell = \alpha, \RS{X}_\ell = \CS{X}}{\RS{S}_\ell = \CS{U}_\ell}\right].~\label{eq:p_rej}
\end{split}
\end{align}
In addition, conditioned on $S_{\ell -1}=\CS{U}_\ell \setminus \left\{c_\ell\right\}$ the event $\bigwedge_{k=qn+1}^{\ell - 1} \neg \ES{M}_\ell$ is independent of $X_\ell$. Therefore, we get
\begin{align}
\begin{split}
    \Pr \left[\given{\bigwedge_{k=qn+1}^{\ell-1} \neg \ES{M}_\ell}{ \RS{S}_{\ell}=\CS{U}_\ell, c_\ell = \alpha, X_\ell = \CS{X} } \right] 
    &=
    \Pr \left[\given{\bigwedge_{k=qn+1}^{\ell-1} \neg \ES{M}_\ell}{ \RS{S}_{\ell-1}=\CS{U}_\ell \setminus \left\{ \alpha\right\} } \right]  \\
    &= \frac{h+qn}{n} \cdot \left(1-\frac{1}{n} \right)^{\ell - 1 - (n-h)},
\end{split}~\label{eq:p_induction}
\end{align}
where the second equality follows from the induction hypothesis. Now by Step~\ref{step 2} and since $X_\ell$ is a uniformly random subset of $\CS{U}_\ell \setminus \left\{c_\ell\right\}$, for every $\alpha \in \CS{U}_\ell$ and $\CS{X} \in \FS{F}^{\alpha}_{n-1}$ we have
 \begin{align}
 \begin{split}
 \Pr\left[ \given{c_\ell = \alpha, \RS{X}_\ell = \CS{X}}{\RS{S}_\ell = \CS{U}_\ell}\right] &= \Pr\left[ \given{c_\ell = \alpha}{\RS{S}_\ell = \CS{U}_\ell}\right] \cdot \Pr\left[ \given{\RS{X}_\ell = \CS{X}}{c_\ell = \alpha, \RS{S}_\ell = \CS{U}_\ell}\right ] \\
 &= \frac{1}{h + \ell} \cdot \frac{1}{\binom{h + \ell -1}{n-1}}.
 \end{split}~\label{eq:p_c_l_x_l}
 \end{align}
Substituting~\eqref{eq:p_induction} and~\eqref{eq:p_c_l_x_l} in Equation \eqref{eq:p_rej}, we get
\begin{align*}
\Pr \left[\given{\bigwedge_{k=qn+1}^{\ell} \neg \ES{M}_\ell}{\RS{S}_\ell = \CS{U}_\ell}\right]
 &= \sum_{\alpha \in \CS{U}_\ell} \sum_{\CS{X} \in \FS{F}_{n-1}^{\alpha}}
    \frac{h+qn}{n} \cdot \left(1-\frac{1}{n} \right)^{\ell-1-(n-h)} 
    \cdot \frac{1}{\ell + h} \cdot \frac{1}{\binom{h + \ell -1}{n-1}} \\
 &= \frac{h+qn}{n} \cdot \left(1-\frac{1}{n} \right)^{\ell-1-(n-h)} 
    \cdot \frac{1}{\ell + h} \cdot \frac{1}{\binom{h + \ell -1}{n-1}}   
    \sum_{\alpha \in \CS{U}_\ell}{\sum_{\CS{X} \in \FS{F}_{n-1}^{\alpha}}{1}}.
 \end{align*}
For a subset $\CS{X} \subseteq \CS{U}_\ell$ let $\CS{X}^{-} = \CS{X} \setminus \max\left\{ \CS{X} \right\}$. By exchange of order of summation, we have $ \sum_{\alpha \in \CS{U}_\ell} \sum_{\CS{X} \in \FS{F}^{\alpha}_{n-1}} {1} = \sum_{\CS{X} \in \FS{F}_n} \sum_{\alpha \in \CS{X}^{-}} 1 = \binom{h + \ell}{n}(n-1) $. Overall, we get
 \begin{align*}
  \Pr \left[\given{\bigwedge_{k=qn+1}^{\ell} \neg \ES{M}_\ell}{\RS{S}_\ell = \CS{U}_\ell}\right]
  &= \frac{h+qn}{n} \cdot \left(1-\frac{1}{n} \right)^{\ell-1-(n-h)} 
    \cdot \frac{n-1}{\ell + h} \cdot \frac{\binom{h + \ell}{n}}{\binom{h + \ell -1}{n-1}} \\
 &= \frac{h+qn}{n} \cdot \left(1-\frac{1}{n} \right)^{\ell-(n-h)}. \qedhere
\end{align*}
\end{proof}

\begin{lemma}~\label{Secretary History Lemma}
For $\ell \geq n - h + 1 $, we have 
\[ 
\E{\RS{R}_\ell} = \frac{h+qn}{n}\cdot\left(1 - \frac{1}{n}\right)^{\ell - (n - h + 1)} \cdot \frac{1}{n}\E{\OPT}.
\]
\end{lemma}
\begin{proof}
By Step~\ref{step 1}, $\RS{S}_\ell \subseteq \CS{C}$ is a uniformly random subset of cardinality $h+\ell$, by Step~\ref{step 2}, $c_\ell$ is uniformly random element of $\RS{S}_\ell$ and by the definition of the algorithm, $\RS{X}_\ell \subseteq \RS{S}_{\ell} \setminus \left\{c_\ell\right\}$ is uniformly random subset of cardinality $n-1$. Therefore, $\RS{X}_\ell \cup \{c_\ell\} \subseteq \CS{C}$ is a uniformly random subset of size $n$. Hence, $\E{ \max\left\{ \RS{X}_\ell \cup \{c_\ell\} \right\} } = \E{\OPT}$. Also, the probability that $c_\ell$ has the maximum value in $\RS{X}_\ell \cup \{c_\ell\} $ is $1/n$. We now bound $\E{\given{R_\ell}{\RS{S}_\ell} = \CS{U}_\ell, \RS{X}_\ell\cup \left\{ c_\ell \right\}=\CS{V}_\ell }$ for any $\CS{U}_\ell \subseteq \CS{C}$ and $\CS{V}_\ell \subseteq \CS{U}_\ell$ of appropriate cardinalities. We have
\begin{align*}
\E{\given{R_\ell}{\RS{S}_\ell} = \CS{U}_\ell, \RS{X}_\ell\cup \left\{ c_\ell \right\}=\CS{V}_\ell } 
&= \max\left\{\CS{V}_\ell\right\} \Pr \left[\given{\ES{M}_\ell \bigwedge_{k=qn+1}^{\ell-1} \neg \ES{M}_k}{S_\ell = \CS{U_\ell}, \RS{X}_\ell\cup \left\{ c_\ell \right\}=\CS{V}_\ell }\right] \\
&= \max\left\{\CS{V}_\ell\right\} \frac{1}{n} \Pr \left[\given{\bigwedge_{k=qn+1}^{\ell-1} \neg \ES{M}_k}{S_{\ell-1} = \CS{U_\ell} \setminus{\max\left\{ \CS{V}_\ell \right\}}}\right]\\
&= \max\left\{\CS{V}_\ell\right\} \frac{1}{n} \cdot \frac{h+qn}{n} \left( 1- \frac{1}{n}\right)^{\ell - \left(n-h+1\right)},
\end{align*}
where the first equality in this derivation is due to the fact that $c_\ell$ is accepted if and only if $c_\ell = \max\left\{\CS{V}_\ell\right\}$, and the last equality follows from Lemma~\ref{Not Accepted Phase2}. We can now conclude
\begin{align*}
    \E{R_\ell} &=  \E{\E{\given{R_\ell}{\RS{S}_\ell}, \RS{X}_\ell\cup \left\{ c_\ell \right\}}} \\
    &= \E{\max\left\{\RS{X}_\ell\cup \left\{ c_\ell \right\} \right\} \frac{1}{n} \cdot \frac{h+qn}{n} \left( 1- \frac{1}{n}\right)^{\ell - \left(n-h+1\right)}}\\
    &=  \frac{h+qn}{n} \cdot \left(1 - \frac{1}{n}\right)^{\ell - (n - h + 1)} \cdot \frac{1}{n}\E{\OPT}. \qedhere
\end{align*}
\end{proof}

\begin{theorem}~\label{secretary CR}
Algorithm~\ref{alg:ro_secretary} is $c(h,n)$-competitive, where \begin{equation*}
    c(h,n) = \begin{cases}
    e^{-e^{-h/n}} &  0 \leq h \leq rn\\
    \frac{h}{n}\left(1-\ln\left(\frac{h}{n}\right)-e^{-h/n}\right) & rn < h \leq n-1,
    \end{cases}
\end{equation*}
and $r \approx 0.567$ is the solution of $e^{-e^{-x}} - x = 0$.
\end{theorem}
Note that for $h=0$, we are in the settings of the ordinary secretary problem. By Theorem~\ref{secretary CR}, the competitive-ratio of Algorithm~\ref{alg:ro_secretary} in this case is $c(0,n)=1/e$. This is not surprising since in this case we set $q(0,n) = 1/e$ and our algorithm is identical to the classical optimal algorithm for the secretary problem. 

\begin{proof}[Proof of Theorem~\ref{secretary CR}]
We use Lemma~\ref{Secretary Lemma} and Lemma~\ref{Secretary History Lemma} to sum over the expected profit of the algorithm in each phase. For Phase~\ref{phase 1}, we get by Lemma~\ref{Secretary Lemma} that
\begin{align*}
\begin{split}
   \sum_{\ell=qn+1}^{n-h}{\E{\RS{R}_\ell}} &= \frac{1}{n}\E{\OPT} \sum_{\ell=qn+1}^{n-h}\frac{h+qn}{h+\ell-1}\\
   &= \E{\OPT}\frac{h+qn}{n}\sum_{k=qn+h}^{n-1}\frac{1}{k} \\
   &\geq \E{\OPT}\frac{h+qn}{n}\ln\left(\frac{n}{qn+h}\right),
\end{split}
\end{align*}
where the inequality follows from the fact that $\sum_{k=qn+h}^{n-1}\frac{1}{k}\geq\intop_{qn+h}^{n}\frac{1}{x}dx=\ln\left(\frac{n}{qn+h}\right)$.
For Phase~\ref{phase 2}, we get by Lemma~\ref{Secretary History Lemma} that
\begin{align*}
\begin{split}
    \sum_{\ell=n-h+1}^{n}{\E{\RS{R}_\ell}} &=\E{\OPT}\frac{1}{n} \cdot \frac{h+qn}{n}\sum_{\ell=n-h+1}^{n}\left(1 - \frac{1}{n}\right)^{\ell - (n - h + 1)}\\
    &= \E{\OPT}\frac{1}{n} \cdot \frac{h+qn}{n}\sum_{k=0}^{h-1}\left(1 - \frac{1}{n}\right)^{k} \\
    &= \E{\OPT}\frac{h+qn}{n}\left(1-\left(1-\frac{1}{n}\right)^h\right)\\
    &\geq \E{\OPT}\frac{h+qn}{n}\left(1-e^{-\frac{h}{n}}\right).
\end{split}
\end{align*}
where the inequality follows from the fact that $1 + x \leq e^x$ for all $x \in \mathbb{R}$. Overall, 
\begin{align*}
    \E{\ALG} &= \sum_{\ell=qn+1}^{n-h}{\E{R_\ell}} + \sum_{\ell=n-h+1}^{n}{\E{R_\ell}} \geq \E{\OPT}\frac{h+qn}{n} \left( \ln\left(\frac{n}{qn+h}\right) +\left(1-e^{-\frac{h}{n}}\right) \right).
\end{align*}
This bound is maximized at $q(h,n)=\max\left\{ e^{-e^{-h/n}} - \frac{h}{n} ,0\right\}$ for which we get $\E{\ALG} \geq c(h,n)\cdot \E{\OPT}$.
\end{proof}
Note that for fixed $n \in \mathbb{N}$, $q(h,n)=\max\left\{ e^{-e^{-h/n}} - \frac{h}{n} ,0\right\}$ is a monotone non-increasing function of $h$, meaning that the sampling phase is getting shorter as the history size grows, and for $h > rn$, the sampling phase is skipped completely.

\begin{theorem}~\label{thm:long_his_sec}
For $h \geq n$, Algorithm~\ref{alg:ro_secretary} is $\left(1-\left(1-\frac{1}{n}\right)^n\right)$-competitive.
\end{theorem}
The proof is similar to the case where $h \leq n-1$. In this case the algorithm starts operating directly from Phase~\ref{phase 2}, thus, a small modification to the proof is need. For completeness, we provide a proof in Appendix~\ref{apx:long_hist}.

We next prove an upper bound on the competitive-ratio for the problem. A similar asymptotic result was established in~\cite{DBLP:conf/ec/CorreaDFS19}, we provide here a proof that applies for any $n\in \mathbb{N}$. 
\begin{theorem}~\label{ro-bound}
Any online algorithm for the $h$-RO-SP has a competitive-ratio of at most 
\begin{equation*}
    \begin{cases}
    \frac{1}{e} \cdot \frac{n+h}{n} + \frac{1}{n}& \frac{h}{n+h} \leq \frac{1}{e} \\
    \frac{h}{n}\ln\left(\frac{h+n}{h}\right) + \frac{1}{n} & \text{\normalfont{otherwise}}.
    \end{cases}
\end{equation*}
\end{theorem}
\begin{proof}
Fix $n,h \in \mathbb{N}$. We consider a classical settings of the secretary problem in which the input to the online player at each online round is only the rank of the arriving candidate among the sub-sequence of already observed candidates, and the goal is to maximize the probability of accepting the best candidate overall. For ease of presentation, we denote this problem by \textit{SP}. Gilbert and Mosteller~\cite{gilbert1966recognizing} showed that the structure of an optimum strategy for the SP is to reject the first $q$-fraction of candidates, for some $q \in [0,1]$, then accepting any candidate who is the best so far. Using this strategy, the probability of accepting the best candidate on input of size $n+h$ is
\[
q\sum_{\ell = q(n+h)}^{n+h-1}{\frac{1}{\ell}} = q\sum_{\ell = q(n+h)+1}^{n+h-1}{\frac{1}{\ell}} +\frac{1}{n+h}
\leq q\intop_{q(n+h)}^{n+h}{\frac{1}{x} dx} +  \frac{1}{n+h}
=  q\ln\left(\frac{1}{q}\right) + \frac{1}{n+h}.
\]
Subject to an additional constraint that the first $h$ candidates must be rejected, i.e., $q \geq \frac{h}{n+h}$, this bound on the probability is maximized for $q=\max\left\{1/e,h/(n+h)\right\}$. Therefore, if the first $h$ candidates must be rejected, the probability of accepting the best candidate is at most
\begin{align*}
    p = \begin{cases}
    \frac{1}{e}+ \frac{1}{n+h}  & \frac{h}{n+h} \leq \frac{1}{e} \\
    \frac{h}{n+h}\ln \left( \frac{h+n}{h}\right) + \frac{1}{n+h} & \text{otherwise}.
    \end{cases}~\label{eq:pr_accept_best}
\end{align*}

Let $0<\varepsilon \ll 1$ and let $\ALG$ be a $c$-competitive online algorithm for the $h$-RO-SP. Following the discussion in the proof of Theorem~\ref{ao-upper-bound}, there exists an infinite $\CS{V} \subseteq \mathbb{N}$, such that $\ALG$ is order-invariant on $\CS{V}$, and we think of $\ALG$ as a family of functions $P_1,\dots,P_n$. We construct an instance for the $h$-RO-SP, $\mathcal{I}=(\CS{C},h)$ such that $\CS{C}= \left\{\alpha_1,\dots, \alpha_{n+h}\right\} \subseteq \CS{V}$, $\alpha_2 > \alpha_3 > \dots > \alpha_{n+h}$ and $\alpha_1 > \alpha_2 / \varepsilon'$ where $\varepsilon' = \frac{n}{n+h} \varepsilon $. Let $p_1$ be the probability that $\ALG$ accepts $\alpha_1$. We get $\E{\ALG} \leq p_1\alpha_1 + \varepsilon' \alpha_1$, whereas $\E{\OPT} \geq \alpha_1 \frac{n}{n+h}$, therefore
\[
c \leq \frac{\E{\ALG(\mathcal{I})}}{\E{\OPT{(\mathcal{I})}}} 
\leq \frac{\alpha_1\left(p_1+\varepsilon'\right)}{\alpha_1\frac{n}{n+h}}
= \frac{n+h}{n}\left(p_1+\varepsilon'\right).
\]

Using $\ALG$, we construct an algorithm $\ALG'$ for the SP subject to the constraint that the first $h$ candidates must be rejected. $\ALG'$ rejects the first $h$ candidates, then, at round $h+j$ for all $j\in[n]$, it constructs a sequence $x_1,\dots,x_{h+j}$ from the elements of $\CS{C}$ which is order-equivalent to the input sequence  of relative ranks until that point in time. Then, it accepts the current candidate with probability $P_j(x_1,\dots,x_{h+j})$. Observe that the probability of $\ALG'$ to accept the best candidate is exactly $p_1$, thus, $p_1 \leq p$. Overall we get 
\[
c 
\leq  \frac{n+h}{n} \left(p_1 + \varepsilon'\right) 
\leq  \frac{n+h}{n} \left(p + \varepsilon'\right) =  \begin{cases}
    \frac{1}{e} \cdot \frac{n+h}{n} + \frac{1}{n} & \frac{h}{n+h} \leq \frac{1}{e} \\
    \frac{h}{n}\ln\left(\frac{h+n}{h}\right) + \frac{1}{n} & \text{\normalfont{otherwise}}
    \end{cases} + \varepsilon. \qedhere
\]
\end{proof}

We note that Theorem~\ref{ro-bound} can also be derived from the work of Buchbinder et al.~\cite{DBLP:journals/mor/BuchbinderJS14}, by including the constraint that the first $h$ candidates must be rejected in the linear program that characterizes all algorithms for the secretary problem.

\subsection{The Relation to the I.I.D. Prophet Inequality with a Sample}

We denote the prophet inequality for $\IID$ random variables from an unknown distribution, with $h$ training samples by \textit{$h$-IID-PI}. We prove the following simple observation in Appendix~\ref{proof:iid_subsummed}. 

\begin{theorem}~\label{thm:iid_subsummed} 
Let $\ALG$ be a $c$-competitive algorithm for the $h$-RO-SP, then $\ALG$ is a $c$-competitive algorithm for the $h$-IID-PI.
\end{theorem}

A direct implication of Theorem~\ref{thm:iid_subsummed} is that the upper bound of $1 / \beta \approx 0.745$ by Hill and Kertz~\cite{hill1982comparisons, kertz1986stop} on the prophet inequality for $\IID$ random variables from a known distribution applies to the $h$-RO-SP for any $h$.\footnote{$\beta$ is the unique value solving $\intop_{0}^{1}{\frac{1}{y\left(1-\ln(y)\right)+(\beta-1)}dy}=1$.}
\begin{corollary}~\label{cor:hill_kertz}
Any online algorithm for the $h$-RO-SP has a competitive-ratio of at most $1/\beta \approx 0.745$.
\end{corollary}

Another consequence of Theorem~\ref{thm:iid_subsummed} is that Algorithm~\ref{alg:ro_secretary} improves upon the results of Correa et al.~\cite{DBLP:conf/ec/CorreaDFS19} for the $h$-IID-PI when $h < n-1$. The improvement is illustrated in Figure~\ref{fig:improvement}.

\begin{figure}
\centering
\caption{Improvement for the $h$-IID-PI}~\label{fig:improvement}
\begin{tikzpicture}[gnuplot]
\path (0.000,0.000) rectangle (11.000,7.000);
\gpcolor{color=gp lt color border}
\gpsetlinetype{gp lt border}
\gpsetlinewidth{1.00}
\draw[gp path] (1.688,0.985)--(1.868,0.985);
\draw[gp path] (9.159,0.985)--(8.979,0.985);
\node[gp node right] at (1.504,0.985) {$0$};
\draw[gp path] (1.688,2.397)--(1.868,2.397);
\draw[gp path] (9.159,2.397)--(8.979,2.397);
\node[gp node right] at (1.504,2.397) {$0.2$};
\draw[gp path] (1.688,3.581)--(1.868,3.581);
\draw[gp path] (9.159,3.581)--(8.979,3.581);
\node[gp node right] at (1.504,3.581) {$1/e$};
\draw[gp path] (1.688,4.514)--(1.868,4.514);
\draw[gp path] (9.159,4.514)--(8.979,4.514);
\node[gp node right] at (1.504,4.514) {$0.5$};
\draw[gp path] (1.688,5.446)--(1.868,5.446);
\draw[gp path] (9.159,5.446)--(8.979,5.446);
\node[gp node right] at (1.504,5.446) {$1-1/e$};
\draw[gp path] (1.688,6.631)--(1.868,6.631);
\draw[gp path] (9.159,6.631)--(8.979,6.631);
\node[gp node right] at (1.504,6.631) {$0.8$};
\draw[gp path] (2.684,0.985)--(2.684,1.165);
\draw[gp path] (2.684,6.631)--(2.684,6.451);
\node[gp node center] at (2.684,0.677) {$0.2$};
\draw[gp path] (3.680,0.985)--(3.680,1.165);
\draw[gp path] (3.680,6.631)--(3.680,6.451);
\node[gp node center] at (3.680,0.677) {$0.4$};
\draw[gp path] (4.676,0.985)--(4.676,1.165);
\draw[gp path] (4.676,6.631)--(4.676,6.451);
\node[gp node center] at (4.676,0.677) {$0.6$};
\draw[gp path] (5.673,0.985)--(5.673,1.165);
\draw[gp path] (5.673,6.631)--(5.673,6.451);
\node[gp node center] at (5.673,0.677) {$0.8$};
\draw[gp path] (6.669,0.985)--(6.669,1.165);
\draw[gp path] (6.669,6.631)--(6.669,6.451);
\node[gp node center] at (6.669,0.677) {$1$};
\draw[gp path] (9.159,0.985)--(9.159,1.165);
\draw[gp path] (9.159,6.631)--(9.159,6.451);
\node[gp node center] at (9.159,0.677) {$1.5$};
\draw[gp path] (1.688,0.985)--(1.688,1.165);
\draw[gp path] (1.688,6.631)--(1.688,6.451);
\node[gp node center] at (1.688,0.677) {$0$};
\draw[gp path] (2.684,0.985)--(2.684,1.165);
\draw[gp path] (2.684,6.631)--(2.684,6.451);
\draw[gp path] (3.680,0.985)--(3.680,1.165);
\draw[gp path] (3.680,6.631)--(3.680,6.451);
\draw[gp path] (4.676,0.985)--(4.676,1.165);
\draw[gp path] (4.676,6.631)--(4.676,6.451);
\draw[gp path] (5.673,0.985)--(5.673,1.165);
\draw[gp path] (5.673,6.631)--(5.673,6.451);
\draw[gp path] (6.669,0.985)--(6.669,1.165);
\draw[gp path] (6.669,6.631)--(6.669,6.451);
\draw[gp path] (9.159,5.446)--(8.979,5.446);
\node[gp node left] at (9.343,5.446) {$1-1/e$};
\draw[gp path] (1.688,6.631)--(1.688,0.985)--(9.159,0.985)--(9.159,6.631)--cycle;
\node[gp node center,rotate=-270] at (0.0,3.808) {competitive-ratio};
\node[gp node center] at (5.423,0.215) {$h/n$};
\node[gp node right] at (8.975,2.682) {Upper bound};
\gpcolor{rgb color={0.000,0.000,0.000}}
\gpsetlinetype{gp lt plot 1}
\gpsetlinewidth{2.00}
\draw[gp path] (4.931,2.682)--(5.847,2.682);
\draw[gp path] (1.688,3.581)--(1.763,3.621)--(1.839,3.660)--(1.914,3.699)--(1.990,3.739)%
  --(2.065,3.778)--(2.141,3.817)--(2.216,3.857)--(2.292,3.896)--(2.367,3.935)--(2.443,3.975)%
  --(2.518,4.014)--(2.594,4.053)--(2.669,4.093)--(2.745,4.132)--(2.820,4.171)--(2.895,4.211)%
  --(2.971,4.250)--(3.046,4.289)--(3.122,4.329)--(3.197,4.368)--(3.273,4.407)--(3.348,4.447)%
  --(3.424,4.486)--(3.499,4.525)--(3.575,4.565)--(3.650,4.604)--(3.726,4.643)--(3.801,4.683)%
  --(3.876,4.722)--(3.952,4.761)--(4.027,4.801)--(4.103,4.840)--(4.178,4.879)--(4.254,4.919)%
  --(4.329,4.958)--(4.405,4.997)--(4.480,5.037)--(4.556,5.076)--(4.631,5.115)--(4.707,5.153)%
  --(4.782,5.191)--(4.858,5.227)--(4.933,5.262)--(5.008,5.296)--(5.084,5.330)--(5.159,5.362)%
  --(5.235,5.394)--(5.310,5.425)--(5.386,5.455)--(5.461,5.485)--(5.537,5.513)--(5.612,5.541)%
  --(5.688,5.569)--(5.763,5.596)--(5.839,5.622)--(5.914,5.648)--(5.989,5.673)--(6.065,5.698)%
  --(6.140,5.722)--(6.216,5.745)--(6.291,5.768)--(6.367,5.791)--(6.442,5.813)--(6.518,5.835)%
  --(6.593,5.856)--(6.669,5.877)--(6.744,5.897)--(6.820,5.917)--(6.895,5.937)--(6.971,5.956)%
  --(7.046,5.975)--(7.121,5.994)--(7.197,6.012)--(7.272,6.030)--(7.348,6.048)--(7.423,6.065)%
  --(7.499,6.082)--(7.574,6.099)--(7.650,6.115)--(7.725,6.131)--(7.801,6.147)--(7.876,6.163)%
  --(7.952,6.178)--(8.027,6.193)--(8.102,6.208)--(8.178,6.223)--(8.253,6.237)--(8.329,6.243)%
  --(8.404,6.243)--(8.480,6.243)--(8.555,6.243)--(8.631,6.243)--(8.706,6.243)--(8.782,6.243)%
  --(8.857,6.243)--(8.933,6.243)--(9.008,6.243)--(9.084,6.243)--(9.159,6.243);
\gpcolor{color=gp lt color border}
\node[gp node right] at (8.975,2.345) { };
\gpcolor{rgb color={1.000,1.000,1.000}}
\gpsetlinewidth{1.00}
\gpsetpointsize{4.00}
\gppoint{gp mark 0}{(5.389,2.345)}
\gpcolor{color=gp lt color border}
\node[gp node right] at (8.975,2.008) {Algorithm~\ref{alg:ro_secretary}};
\gpcolor{rgb color={0.000,0.000,0.000}}
\gpsetlinetype{gp lt plot 0}
\gpsetlinewidth{3.00}
\draw[gp path] (4.931,2.008)--(5.847,2.008);
\draw[gp path] (1.688,3.581)--(1.763,3.621)--(1.839,3.660)--(1.914,3.699)--(1.990,3.739)%
  --(2.065,3.778)--(2.141,3.817)--(2.216,3.856)--(2.292,3.895)--(2.367,3.934)--(2.443,3.973)%
  --(2.518,4.012)--(2.594,4.051)--(2.669,4.090)--(2.745,4.128)--(2.820,4.167)--(2.895,4.205)%
  --(2.971,4.243)--(3.046,4.281)--(3.122,4.319)--(3.197,4.357)--(3.273,4.395)--(3.348,4.432)%
  --(3.424,4.470)--(3.499,4.507)--(3.575,4.544)--(3.650,4.581)--(3.726,4.617)--(3.801,4.654)%
  --(3.876,4.690)--(3.952,4.726)--(4.027,4.762)--(4.103,4.798)--(4.178,4.833)--(4.254,4.868)%
  --(4.329,4.903)--(4.405,4.938)--(4.480,4.973)--(4.556,5.007)--(4.631,5.040)--(4.707,5.071)%
  --(4.782,5.101)--(4.858,5.129)--(4.933,5.156)--(5.008,5.182)--(5.084,5.206)--(5.159,5.230)%
  --(5.235,5.251)--(5.310,5.272)--(5.386,5.291)--(5.461,5.309)--(5.537,5.326)--(5.612,5.342)%
  --(5.688,5.357)--(5.763,5.370)--(5.839,5.383)--(5.914,5.394)--(5.989,5.404)--(6.065,5.413)%
  --(6.140,5.421)--(6.216,5.428)--(6.291,5.433)--(6.367,5.438)--(6.442,5.442)--(6.518,5.444)%
  --(6.593,5.446)--(6.669,5.446)--(6.744,5.446)--(6.820,5.446)--(6.895,5.446)--(6.971,5.446)%
  --(7.046,5.446)--(7.121,5.446)--(7.197,5.446)--(7.272,5.446)--(7.348,5.446)--(7.423,5.446)%
  --(7.499,5.446)--(7.574,5.446)--(7.650,5.446)--(7.725,5.446)--(7.801,5.446)--(7.876,5.446)%
  --(7.952,5.446)--(8.027,5.446)--(8.102,5.446)--(8.178,5.446)--(8.253,5.446)--(8.329,5.446)%
  --(8.404,5.446)--(8.480,5.446)--(8.555,5.446)--(8.631,5.446)--(8.706,5.446)--(8.782,5.446)%
  --(8.857,5.446)--(8.933,5.446)--(9.008,5.446)--(9.084,5.446)--(9.159,5.446);
\gpcolor{color=gp lt color border}
\node[gp node right] at (8.975,1.671) { };
\gpcolor{rgb color={1.000,1.000,1.000}}
\gpsetlinewidth{1.00}
\gppoint{gp mark 0}{(5.389,1.671)}
\gpcolor{color=gp lt color border}
\node[gp node right] at (8.975,1.334) {Correa et al.~\cite{DBLP:conf/ec/CorreaDFS19}};
\gpcolor{rgb color={0.745,0.745,0.745}}
\gpsetlinetype{gp lt plot 2}
\gpsetlinewidth{3.00}
\draw[gp path] (4.931,1.334)--(5.847,1.334);
\draw[gp path] (1.688,3.581)--(1.763,3.581)--(1.839,3.581)--(1.914,3.581)--(1.990,3.581)%
  --(2.065,3.581)--(2.141,3.581)--(2.216,3.581)--(2.292,3.581)--(2.367,3.581)--(2.443,3.581)%
  --(2.518,3.587)--(2.594,3.621)--(2.669,3.655)--(2.745,3.689)--(2.820,3.723)--(2.895,3.756)%
  --(2.971,3.790)--(3.046,3.824)--(3.122,3.858)--(3.197,3.892)--(3.273,3.925)--(3.348,3.959)%
  --(3.424,3.993)--(3.499,4.027)--(3.575,4.061)--(3.650,4.094)--(3.726,4.128)--(3.801,4.162)%
  --(3.876,4.196)--(3.952,4.230)--(4.027,4.263)--(4.103,4.297)--(4.178,4.331)--(4.254,4.365)%
  --(4.329,4.398)--(4.405,4.432)--(4.480,4.466)--(4.556,4.500)--(4.631,4.534)--(4.707,4.567)%
  --(4.782,4.601)--(4.858,4.635)--(4.933,4.669)--(5.008,4.703)--(5.084,4.736)--(5.159,4.770)%
  --(5.235,4.804)--(5.310,4.838)--(5.386,4.872)--(5.461,4.905)--(5.537,4.939)--(5.612,4.973)%
  --(5.688,5.007)--(5.763,5.041)--(5.839,5.074)--(5.914,5.108)--(5.989,5.142)--(6.065,5.176)%
  --(6.140,5.210)--(6.216,5.243)--(6.291,5.277)--(6.367,5.311)--(6.442,5.345)--(6.518,5.379)%
  --(6.593,5.412)--(6.669,5.446)--(6.744,5.446)--(6.820,5.446)--(6.895,5.446)--(6.971,5.446)%
  --(7.046,5.446)--(7.121,5.446)--(7.197,5.446)--(7.272,5.446)--(7.348,5.446)--(7.423,5.446)%
  --(7.499,5.446)--(7.574,5.446)--(7.650,5.446)--(7.725,5.446)--(7.801,5.446)--(7.876,5.446)%
  --(7.952,5.446)--(8.027,5.446)--(8.102,5.446)--(8.178,5.446)--(8.253,5.446)--(8.329,5.446)%
  --(8.404,5.446)--(8.480,5.446)--(8.555,5.446)--(8.631,5.446)--(8.706,5.446)--(8.782,5.446)%
  --(8.857,5.446)--(8.933,5.446)--(9.008,5.446)--(9.084,5.446)--(9.159,5.446);
\gpcolor{color=gp lt color border}
\gpsetlinetype{gp lt border}
\gpsetlinewidth{1.00}
\draw[gp path] (1.688,6.631)--(1.688,0.985)--(9.159,0.985)--(9.159,6.631)--cycle;
\gpdefrectangularnode{gp plot 1}{\pgfpoint{1.688cm}{0.985cm}}{\pgfpoint{9.159cm}{6.631cm}}
\end{tikzpicture}
`
\end{figure}
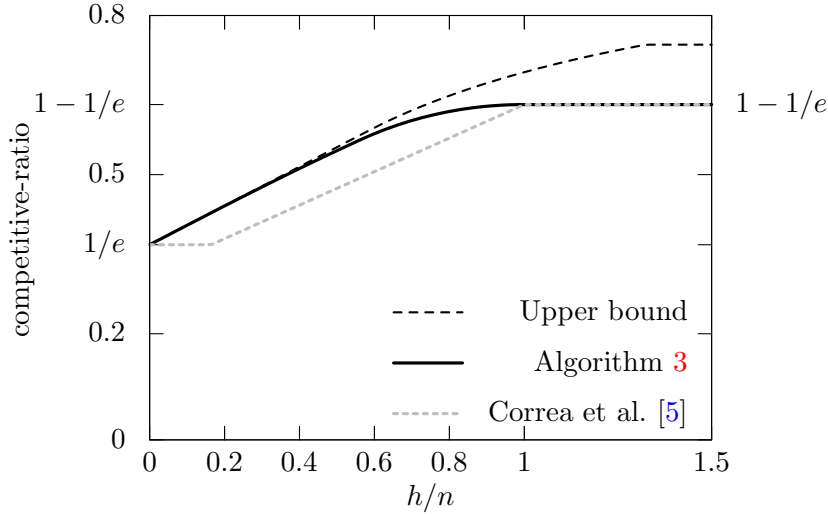

\section{Have it Both Ways}\label{sec:combine}
Clearly, any performance guarantee in the $\AOS$ model applies also to the $\ROS$ model.
A natural question in this context is: can we have a single algorithm that performs well in both models? A very desired property of an algorithm would be a good worst-case performance guarantee, and a better one in case the input arrives in a random order. Next we formalize and prove the following intuitive result: optimizing the algorithm for the worst-case might inherently reduce performance in the random-order case, and vice versa. 

\begin{theorem}\label{thm:wc-ro}
Let $\ALG$ be a $c$-competitive algorithm for the $h$-AO-SP, and let $\varepsilon>0$. Then there exists $n_0\in \mathbb{N}$ such that for any $n\geq n_0$, the competitive-ratio of $\ALG$ for the $h$-RO-SP on instances of length $n+h$ is at most $(1-c+\varepsilon)$.
\end{theorem}

\begin{proof}
Fix $\varepsilon > 0$. In  a similar way to the proof of Theorem~\ref{thm:half_upper_bound}, we construct two instances: the first instance $\mathcal{I}_1$ consists of a set $\CS{E}$ of $m$ candidates of value $\varepsilon' = \varepsilon/2$ , and a set $\CS{Z}$ of $n+h-m$ candidates of value $0$. For the second instance $\mathcal{I}_2$, we replace one arbitrary candidate of value $0$ with a candidate $\alpha$ of value $\frac{n+h}{n}$. 

Consider an adversary in the $\AOS$ model which first reveals all candidates in the online set except $\alpha$ in an arbitrary order. At the end, it reveals $\alpha$ if $\alpha \in \RS{O}$. Now, since $\ALG$ is $c$-competitive in the $\AOS$ model, as in the proof of Theorem~\ref{thm:half_upper_bound}, conditioned on $\alpha \in \RS{O}$, the probability that $\ALG$ accepts $\alpha$ is at least $c-\varepsilon'$. Since $\alpha$ arrives last, it follows that $\Pr\left[\given{\ALG(\mathcal{I}_2)\text{ reaches }c_n}{\alpha \in O}\right] \geq c-\varepsilon'$, and since it replaces a uniformly random element of $\CS{Z}$ from $\mathcal{I}_1$, until arriving to the last candidate, it is impossible to distinguish between facing $\mathcal{I}_2$ conditioned on $\alpha \in O$, and facing $\mathcal{I}_1$ conditioned on $c_n \in \CS{Z}$. Therefore, 
\begin{align}
    \Pr\left[\given{\ALG(\mathcal{I}_1)\text{ reaches }c_n}{c_n \in \CS{Z}}\right] \geq c - \varepsilon'.\label{eq:pr_reaches_c_n}
\end{align}
Note that the above probabilities are taken over the random choice of $\RS{O} \subseteq \CS{C}$, and the internal randomness of $\ALG$.

We now move to analyze the expected profit of the algorithm in the $\ROS$ model. We have
\begin{align}
 \E{\ALG(\mathcal{I}_1)} 
 &= \E{\given{\ALG(\mathcal{I}_1)}{c_n \in \CS{Z}}}\Pr\left[c_n \in \CS{Z}\right] 
 + \E{\given{\ALG(\mathcal{I}_1)}{c_n \notin \CS{Z}}}\Pr\left[c_n \notin \CS{Z}\right].\label{eq:ex_I_1}
\end{align}
Conditioned on $c_n \in \CS{Z}$, when the algorithm reaches $c_n$, its profit is $0$. Otherwise, its profit is at most $\varepsilon'$. Therefore, we get that
\begin{align}
\E{\given{\ALG(\mathcal{I}_1)}{c_n \in \CS{Z}}}
 &\leq \varepsilon' \left(1 - \Pr\left[\given{\ALG(\mathcal{I}_1) \text{ reaches }c_n}{c_n \in \CS{Z}}\right] \right) \leq \varepsilon' \left(1 - \left(c - \varepsilon' \right) \right),
 \label{eq:c_n_z}
\end{align}
where the last inequality follows from~\eqref{eq:pr_reaches_c_n}, since it applies for any order of the elements in $\RS{O}$, it also applies when they randomly ordered. In addition, we clearly have $\E{\given{\ALG(\mathcal{I}_1)}{c_n \notin \CS{Z}}}\leq \varepsilon'$. Using this and Inequality \eqref{eq:c_n_z} in \eqref{eq:ex_I_1} we get that
\begin{align*}
 \E{\ALG(\mathcal{I}_1)} 
 &\leq \varepsilon' \left(1 - \left(c - \varepsilon' \right) \right)\Pr[c_n \in \CS{Z}]
  + \varepsilon'\Pr[c_n \notin \CS{Z}] \\
 &= \varepsilon'\left(1-\frac{n+h-m}{n+h}(c-\varepsilon')\right).
\end{align*}
On the other hand, $  \E{\OPT(\mathcal{I}_1)} = \varepsilon'\Pr[\CS{E} \cap O \neq \emptyset]$. We have
\begin{align*}
\frac{\E{\ALG(\mathcal{I}_1)}}{\E{\OPT(\mathcal{I}_1)}} 
& \leq \frac{1-\left(1-\frac{m}{n+h} \right)(c  - \varepsilon')}{\Pr[\CS{E} \cap O \neq \emptyset]}.
\end{align*}
Using Proposition~\ref{prop:choose-r} with $r=\frac{n+h}{n+h-m}$, we can bound $\Pr[\CS{E} \cap O = \emptyset] \leq \left( 1- \frac{m}{n+h} \right)^n \leq e^{-\frac{n\cdot m}{n+h}}$. By choosing $m = \frac{n+h}{n}\ln(n)$, we get that 
\begin{align*}
\frac{\E{\ALG(\mathcal{I}_1)}}{\E{\OPT(\mathcal{I}_1)}}  \leq \frac{n}{n-1}
\left( 1-\left(1- \frac{\ln(n)}{n} \right)(c-\varepsilon') \right) \xrightarrow[n \to \infty]{} 1-(c-\varepsilon').
\end{align*}
Therefore, there exists $n_0 \in \mathbb{N}$ such that for any $n \geq n_0$, we have ${\E{\ALG(\mathcal{I}_1)}}/{\E{\OPT(\mathcal{I}_1)}}  \leq 1-(c-\varepsilon') + \varepsilon' = 1-c+\varepsilon$.
\end{proof}

Note that a weaker (asymptotic) version of Theorem~\ref{thm:half_upper_bound} follows from Theorem~\ref{thm:wc-ro} as follows. Assume by contradiction that $\ALG$ is $(1/2+\varepsilon)$-competitive in the $\AOS$ model and therefore also in the $\ROS$ model. Then by Theorem~\ref{thm:wc-ro}, $\ALG$ is at most $1/2$-competitive in the $\ROS$ model and we get a contradiction.
It also follows that Algorithm~\ref{alg:adv-sec-long-hist} is an optimal $1/2$-competitive algorithm when considering both models. A natural question that arises here is, are there other values of $c$ for which we can obtain a $c$-competitive algorithm in the $\AOS$ model which is $(1-c)$-competitive in the $\ROS$ model? We answer this question in the affirmative for $c=1/e$ and $h \geq n(n-1)$. 

\begin{theorem}
For $h \geq n(n-1)$, there exists a $1/e$-competitive algorithm for the $h$-AO-SP which is also $(1-1/e)$-competitive for the $h$-RO-SP.
\end{theorem}

\begin{proof}
Consider an algorithm that draws a uniformly random subset of cardinality $n(n-1)$ from the history set, and partitions it into $n$ uniformly random subsets $S_1,\dots,S_n$ of cardinality $n-1$ each. Then, at each online round $\ell \in [n]$, the algorithm accepts the current candidate $c_\ell$ if and only if $c_\ell > \max\left\{S_\ell\right\}$.

Let $x_1, \dots, x_n$ be the elements of $\RS{O}$ in a uniformly random order. If we randomly pick a subset $\RS{X} \subseteq \RS{H}$ of cardinality $n(n-1)$, and randomly partition it into $n$ subsets $X_1,\dots,X_n$  of cardinality $n-1$ each, then, $\RS{T}=\left((X_1, x_1), \dots, (X_n,x_n)\right)$ is a tuple of pairs consisting of $n^2$ distinct elements from $\CS{C}$. Each such tuple has the same probability to be chosen. Therefore, $X_i \cup \left\{x_i\right\} \subseteq \CS{C}$ is a uniformly random subset of cardinality $n$, and $x_i > \max\{X_i\}$ with probability $1/n$ independently for each $i\in [n]$. Let $\RS{P}\subseteq \RS{O}$ be the set of candidates that satisfy $x_i > \max \left\{X_i\right\}$.
If the candidates of $\RS{O}$ arrive in random order, the algorithm that we suggest processes the elements of $T$ in a random order, and therefore, it accepts a random element of $\RS{P}$. When it accepts a candidate, it is the maximum of a uniformly random subset of cardinality $n$ from $\CS{C}$, thus, its profit in this case is $\E{\OPT}$. The probability that the algorithm accepts a candidate (i.e., $P\neq \emptyset$) is 
$1- \left(1-1/n\right)^{n} \geq 1-1/e$. Therefore, in the $\ROS$ model, the profit of the algorithm is at least $\left(1-1/e\right)\E{\OPT}$.

When the elements of $\RS{O}$ arrive in adversarial order, the adversary can force the algorithm to pick the worst element in $\RS{P}$. Observe that in case $\left|P\right|=1$, the arrival order is irrelevant and the algorithm must accept the only candidate in $P$. We get that the profit of the algorithm in the adversarial-order case, is at least the profit of the algorithm in the random-order case when $\left|P\right| = 1$. We therefore return to the random-order case and lower bound the profit of the algorithm in the $\AOS$ model by using the above observation. Following the discussion in the previous paragraph, the probability that $x_i > \max\{X_i\}$ for exactly one $i\in[n]$ (i.e., $\left|P\right|=1$) is $\binom{n}{1}\frac{1}{n}\left(1-\frac{1}{n}\right)^{n-1}$ which for $n > 1$ is at least $1/e$ (the case $n=1$ is trivial). In this case the profit of the algorithm is $\E{\OPT}$. Overall, the expected profit of the algorithm in the $\AOS$ model is at least $\frac{1}{e}\E{\OPT}$.

\end{proof}

We note that Algorithm~\ref{alg:ro_secretary} does not achieve any bounded competitive-ratio in the $\AOS$ model when a sampling phase is used (i.e., for $h < 0.567n$), since the adversary can ensure the algorithm does not accept the best candidate by placing him in the sampling phase, and any other choice might yield a negligible profit. 

\section{Discussion}\label{sec:discussion}
In this paper we introduce new models for the design and analysis of online algorithms, where prior knowledge can be accounted for while preserving most of the power of the adversary. This is done by making the adversary expose a random sample of a larger worst case input in advance to the online player. The adversary then uses the remaining part of the input at the online stage. In this way, we model ``similar'' data that the online player can learn from in advance.

We note that our models also cover other natural settings. A particularly interesting example is the following: assume a finite population of $m$ elements is chosen by an adversary, such that any sample of $n$ elements from the population defines an input instance. An online player gets to sample $h$ elements from the population (without replacement) for learning. Then, the adversary gets to sample $n$ elements and challenges the player in an online fashion with these $n$ samples (in adversarial order or random order).

Many interesting questions are open for future research. For the secretary problem in the $\ROS$ model, a gap between the lower and upper bounds remains unresolved. 
There is also room for designing online algorithms with the objective of optimizing the performance in both models simultaneously.

As mentioned in the introduction the models we describe here are general and can be applied to various online problems. A particularly interesting example is the weighted bipartite matching problem. Following the approach of~\cite{DBLP:conf/esa/KesselheimRTV13}, Algorithm~\ref{alg:ro_secretary} can be easily extended to the weighted bipartite matching problem in the $\ROS$ model, without any loss in the competitive-ratio. An interesting question is, can the same be done in the $\AOS$ model?

\appendix

\section{Omitted Proofs}

\subsection{Proof of Proposition~\ref{prop:choose-r}}\label{apx:choose-r}
\begin{proofw}
\begin{align*}
\binom{n}{k} 
&= \frac{n \cdot \left(n -1\right) \cdots \left(n - (k-1)\right)}{k!} \\
&= \frac{1}{r^k}  \cdot \frac{r \cdot n \cdot r \cdot \left(n -1\right)  \cdots r \cdot \left(n - (k-1)\right)}{k!} \\
&\leq \frac{1}{r^k}  \cdot \frac{(r n) \cdot \left(r n-1\right) \cdots  \left(r n - (k-1)\right)}{k!} \\
&= \frac{1}{r^k}\binom{r n}{k}. \qedhere
\end{align*}
\end{proofw}

\subsection{Proof of Proposition~\ref{expected max prop}}\label{proof expected max prop}
\begin{proofw}
Let $\RS{C}\subseteq \RS{B}$ be a uniformly random subset of cardinality $k$. Since $\RS{B}$ is a uniformly random subset of $\CS{X}$ of cardinality $n$, $\RS{C}$ is also a uniformly random subset of $\CS{X}$ of cardinality $k$. Therefore $\E{\max\left\{\RS{A}\right\}}=\E{\max\left\{\RS{C}\right\}}$. By law of total expectation, we have
\begin{align*}
    \E{\max\left\{ \RS{C} \right\}}
                &=     \E{\E{\given{\max\left\{ \RS{C}\right\}}{ \RS{B} }}}
                \geq  \E{ \max\left\{ \RS{B}\right\}\Pr \left[ \given{\max\left\{\RS{B}\right\} \in \RS{C}}{ \RS{B}} \right] } \\
                &=  \E{\frac{k}{n}\max\left\{\RS{B}\right\}} 
                =     \frac{k}{n}  \E{\max\left\{\RS{B}\right\}}.\qedhere
\end{align*}

\end{proofw}

\subsection{Proof of Theorem~\ref{thm:long_his_sec}}\label{apx:long_hist}
\begin{proofw}
Following the inductive argument in the proof of Lemma~\ref{Not Accepted Phase2} and replacing the base case to reflect the fact that the algorithm starts directly from Phase~\ref{phase 2}, it is easy to verify that for every round $\ell\in [n]$ and for any $\CS{U}_\ell\subseteq \CS{C}$ such that $\left|\CS{U}_\ell\right|= h+\ell$, we have
\begin{align*}
 \Pr \left[\given{\bigwedge_{k=1}^{\ell} \neg \ES{M}_k}{\RS{S_\ell}=\CS{U}_\ell}\right] =  \left(1-\frac{1}{n} \right)^{\ell}.~\label{eq:not_accepted_long_hist}
\end{align*}
Following the proof of Lemma~\ref{Secretary History Lemma} and replacing the use of Lemma~\ref{Not Accepted Phase2} by the above result, we get that for every $\ell \in [n]$
\[ 
\E{\RS{R}_\ell} = \left(1 - \frac{1}{n}\right)^{\ell - 1} \cdot \frac{1}{n}\E{\OPT}.
\]
Now we can sum over profit of the algorithm at each round and get the theorem.
\end{proofw}

\subsection{Proof of Theorem~\ref{thm:iid_subsummed}}\label{proof:iid_subsummed}
\begin{proofw}
Let $T_1,\dots,T_h, X_1,\dots, X_n$ be $\IID$ random variables drawn from a distribution $F$ where $T_1,\dots, T_h$ denote are the training samples, and $X_i$ is the sample that arrives at round $i$, for $i\in[n]$. We may assume that the realizations of $T_1,\dots, T_h, X_1,\dots, X_n$ are obtained by first independently drawing $n+h$ values from $F$ to obtain a sequence of values $C = (Y_1,\dots, Y_{n+h})$, then permuting them uniformly at random to obtain the value of each random variable. Fix $\CS{C}=(y_1,\dots,y_{h+n})$. Conditioned on $C = \CS{C}$, the problem is identical to the $h$-RO-SP, since $T_1, \dots, T_h$ are $h$ uniformly random elements of $\CS{C}$, and the remaining elements arrive in a uniformly random order. Therefore, we get
\[
\E{\given{ \ALG \left( \left\{ T_1,\dots,T_h, X_1,\dots, X_n\right\},h \right) } {C=\CS{C}}}
\geq c \cdot \E{\given{\max\left\{X_1,\dots, X_n\right\}}{C=\CS{C}}},
\]
where the expectation is taken over the random permutation of the elements in $\CS{C}$. The lemma follows by taking the expectation over $C$, or explicitly
\begin{align*}
\E{\ALG \left( \left\{ T_1,\dots,T_h, X_1,\dots, X_n\right\},h \right)}
&= \E{\E{\given{ \ALG \left( \left\{ T_1,\dots,T_h, X_1,\dots, X_n\right\},h \right) } {C}}} \\
&\geq \E{ c \cdot \E{\given{\max\left\{X_1,\dots, X_n\right\}}{C}}} \\
&= c \cdot \E{\max\left\{X_1,\dots, X_n\right\}}. \qedhere
\end{align*}
\end{proofw}
\bibliographystyle{plain}
\bibliography{citations}

\end{document}